\documentclass[twoside,11pt]{article}

\usepackage[preprint]{jmlr2e}

\usepackage{color}

\usepackage[utf8]{inputenc} 
\usepackage[T1]{fontenc}    
\usepackage{booktabs}       
\usepackage{amsfonts}       
\usepackage{nicefrac}       
\usepackage{microtype}      
\usepackage{xcolor}         

\usepackage{hyperref}
\hypersetup{colorlinks = true,
            linkcolor = blue,
            urlcolor  = blue,
            citecolor = blue,
            anchorcolor = blue}
\usepackage[hyphenbreaks]{breakurl}
\usepackage{xurl}

\usepackage{graphicx}
\usepackage{mathtools,bbm,nicefrac}
\usepackage{color}
\usepackage{braket}
\usepackage{algorithm, algorithmic}
\usepackage[whole]{bxcjkjatype}
\allowdisplaybreaks[4]

\DeclareMathOperator{\tr}{Tr}

\DeclareMathOperator{\poly}{poly}
\DeclareMathOperator{\polylog}{polylog}

\DeclareMathOperator*{\argmin}{argmin}

\DeclareMathOperator{\sign}{sign}
\DeclareMathOperator{\supp}{supp}

\newcommand{\cX}{\mathcal{X}}

\newcommand{\cR}{\mathcal{R}}
\newcommand{\cL}{\mathcal{L}}

\jmlrheading{}{}{}{}{}{}{Hayata Yamasaki and Sho Sonoda}

\ShortHeadings{Exponential Error Convergence in Classfication with Optimized Random Features}{Yamasaki and Sonoda}
\firstpageno{1}

\begin{document}

\title{Exponential Error Convergence in Data Classification with Optimized Random Features: Acceleration by Quantum Machine Learning}

\author{\name Hayata Yamasaki \email hayata.yamasaki@gmail.com \\
  \addr IQOQI Vienna, Austrian Academy of Sciences,\\
  Boltzmanngasse 3, 1090 Vienna, Austria\\
  \addr Atominstitut, Technische Universit\"at Wien,\\
  Stadionallee 2, 1020 Vienna, Austria
       \AND
       \name Sho Sonoda \email sho.sonoda@riken.jp \\
       \addr RIKEN AIP,\\
       Nihonbashi 1--4--1, Chuo-ku, Tokyo, Japan}

\editor{}

\maketitle

\begin{abstract}
Classification is a common task in machine learning.
Random features (RFs) stand as a central technique for scalable learning algorithms based on kernel methods,
and more recently proposed \textit{optimized random features}, sampled depending on the model and the data distribution, can significantly reduce and provably minimize the required number of features.
However, existing research on classification using optimized RFs has suffered from computational hardness in sampling each optimized RF\@; moreover, it has failed to achieve the exponentially fast error-convergence speed that other state-of-the-art kernel methods can achieve under a low-noise condition.
To overcome these slowdowns, we here construct a classification algorithm with optimized RFs accelerated by means of quantum machine learning (QML) and study its runtime to clarify overall advantage.
We prove that our algorithm can achieve the exponential error convergence under the low-noise condition even with optimized RFs\@; at the same time, our algorithm can exploit the advantage of the significant reduction of the number of features without the computational hardness owing to QML\@.
These results discover a promising application of QML to acceleration of the leading kernel-based classification algorithm without ruining its wide applicability and the exponential error-convergence speed.
\end{abstract}


\section{\label{sec:intro}Introduction}

\textit{Background}. ---
Classification deals with obtaining a classifier function $f$ to label input data as correctly as possible, from given $N$ examples of input data and their labels.
Kernel methods are a class of widely used methods for solving such a task, owing to their theoretical guarantee and good accuracy~\citep{S5,W2,Steinwart2008}.
However, typical kernel methods, which compute an $N\times N$ Gram matrix for $N$ examples, are not scalable as $N$ gets large.
To scale up kernel methods to big data, random features (RFs) proposed by~\citet{R2,R3} are one of the most highly appreciated techniques of central use in practice, along with other techniques based on low-rank matrix approximation~\citep{10.5555/645529.657980,10.5555/3008751.3008847,10.5555/944790.944812}.
Algorithms with random Fourier features use the fact that any translation-invariant kernel $k:\cX\times\cX\to\mathbb{R}$ on $D$-dimensional input data space $\cX=\mathbb{R}^D$, e.g., a Gaussian kernel, can be represented as an expectation of feature maps $\varphi(v,x)\coloneqq\mathrm{e}^{-2\pi\mathrm{i}v\cdot x}$, i.e.,
\begin{equation}
  k(x,x^\prime)=\int_\mathcal{V}d\tau(v)\overline{\varphi(v,x)}\varphi(v,x^\prime),
\end{equation}
where $\varphi:\mathcal{V}\times\mathcal{X}\to\mathbb{C}$, and $\mathcal{V}=\mathbb{R}^D$ is a parameter space equipped with a probability measure $d\tau(v)$ given by the Fourier transform of $k$; e.g., $d\tau(v)$ of the Gaussian kernel is a Gaussian distribution~\citep{R2}.
With RFs, we represent the non-linear function $f:\mathcal{X}\to\mathbb{R}$ as a linear combination of non-linear feature maps, i.e.,  $\cos$ and $\sin$,
\begin{align}
  \label{eq:estimate}
  &f(x)\approx \hat{f}_{v,\alpha}(x)\coloneqq\sum_{m=0}^{M-1}(\alpha_{2m}\cos(-2\pi v_m\cdot x)+\alpha_{2m+1}\sin(-2\pi v_m\cdot x)).
\end{align}
Learning with RFs is achieved by sampling many feature maps at random parameters $v_0,\ldots,v_{M-1}\in\mathcal{V}$, followed by finding appropriate coefficients $\alpha_0,\ldots,\alpha_{2M-1}\in\mathbb{R}$ by convex optimization using the $N$ examples.
Conventionally, $v_0,\ldots,v_{M-1}$ are sampled from a \textit{data-independent} probability distribution $d\tau(v)$ depending only on the kernel $k$, but this may \textit{require a large number of features} $M=\widetilde{O}(\nicefrac{1}{\epsilon^2})$ for approximating $f$ to accuracy $\epsilon$~\citep{R2,R3}, where $\widetilde{O}$ ignores poly-logarithmic factors.
The requirement of large $M$ \textit{slows down} the decision of all $M$ features, the regression over the $2M$ coefficients, and the evaluation of the learned function $\hat{f}_{v,\alpha}(x)$ in using it after the learning.
Acceleration of kernel methods with RFs is not a specific problem but of central importance for their various applications including the classification.

To achieve such acceleration,~\citet{B1} has proposed to sample $v_0,\ldots,v_{M-1}$ from a \textit{data-optimized} probability distribution that puts greater weight on important features optimized for the data.
The optimized distribution is given by a weighted distribution $q_\lambda^\ast(v)d\tau(v)$, where $q_\lambda^\ast(v)$ will be defined later in~\eqref{eq:q_complex}.
The function $q_\lambda^\ast(v)$ is called the leverage score.
We call features sampled from this data-optimized distribution \textit{optimized random features}.
The use of $q_\lambda^\ast(v)d\tau(v)$ can significantly reduce the required number $M$ of RFs, which is provably optimal up to a logarithmic gap;
e.g., for the Gaussian kernel and a sub-Gaussian data distribution, the required number of optimized RFs for accuracy $\epsilon$ is as small as $M=O(\polylog(\nicefrac{1}{\epsilon}))$~\citep{B1}, which is exponentially smaller than $M=\widetilde{O}(\nicefrac{1}{\epsilon^2})$ of conventional RFs~\citep{R2,R3}.
For regression tasks,~\citet{B1,NIPS2017_6914} had clarified the generalization property of optimized RFs, i.e., the required number $N$ of examples and $M$ of RFs;
based on these works,~\citet{NEURIPS2020_9ddb9dd5} has studied the advantage of the optimized RFs in reducing the overall runtime for the regression.
However, for the classification, the runtime advantage of the optimized RFs has been unknown.

\textit{Problem and results}. ---
In this work, we propose a total design of the setting and the algorithm to use the optimized RFs for a classification task, so as to prove significant runtime advantage of the optimized RFs in the classification.
Progressing beyond regression,~\citet{NIPS2018_7598} has analyzed the generalization property of optimized RFs for classification under a low-noise condition, which is characterized by a constant $\delta>0$ in such a way that larger $\delta$ implies lower noise (as shown later in~\eqref{eq:low_noise}).
Following~\citet{NIPS2018_7598}, we here study a classification task under the low-noise condition with the Gaussian kernel.
Given a desired error $\epsilon>0$,
the task aims to obtain a classifier in the form of $\hat{f}_{v,\alpha}$ in~\eqref{eq:estimate} that reduces the excess classification error to $\epsilon$ compared with the optimal one, as formulated later in~\eqref{eq:excess_testing_error}.
For data sets given according to sub-Gaussian data distributions,~\citet{NIPS2018_7598} has shown that optimized RFs lead to the exponential reduction of the required number $M$ of RFs, similarly to the cases of regression.

However, problematically, the reduction of the number $M$ of RFs does not necessarily imply the speedup in overall runtime since there can be other computational bottlenecks in the algorithm.
Indeed, optimized RFs have been hard to compute in practice, due to the exponential runtime $O(\exp(D))$ in sampling each optimized RF from $q_\lambda^\ast(v)d\tau(v)$ by the classical algorithms~\citep{B1,NIPS2018_7598,Shahrampour2019}; regarding this hardness,~\citet{NIPS2018_7598} does not provide solution to achieve feasible overall runtime.
A possible solution to this bottleneck has recently arisen in the emerging field of \textit{quantum machine learning} (QML), i.e., machine learning with assistance of quantum computer for acceleration~\citep{biamonte2017quantum,doi:10.1098/rspa.2017.0551,dunjko2018machine}.
In particular,~\citet{NEURIPS2020_9ddb9dd5} has established a quantum algorithm, run by quantum computer, to achieve the sampling from $q_\lambda^\ast(v)d\tau(v)$ feasibly in linear runtime $O(D)$, achieving an exponential speedup in $D$ compared to the classical algorithms.
But this quantum algorithm is designed for the regression, and how we can use it for the classification without canceling out the speedup has been unknown.
In addition, for $N$ examples,~\citet{NIPS2018_7598} only showed an error bound scaling polynomially, i.e., $\epsilon=O(\poly(\nicefrac{1}{N}))$ in $N$; after all, the optimized RFs may improve the scaling in the number $M$ of RFs rather than $N$.
This bound in $N$ is exponentially worse than the other state-of-the-art kernel methods for the same classification task under the low-noise condition, which may not use optimized RFs to reduce $M$ but can achieve the exponentially better bound $\epsilon=O(\exp(-N))$ in $N$~\citep{pmlr-v75-pillaud-vivien18a,pmlr-v89-nitanda19a,Yashima2019}.

To address this problem, we here develop a significantly fast algorithm for the classification task using the optimized RFs (Algorithm~\ref{alg:classification}) and quantitatively bound its runtime to prove the runtime advantage of this algorithm without degradation in both $N$ and $M$.
To show the fast runtime, we clarify the appropriate setting for sampling optimized RFs based on the quantum algorithm of~\citet{NEURIPS2020_9ddb9dd5}, and after the sampling, we determine the coefficients of the classifier as efficiently as possible using stochastic gradient descent (SGD)~\citep{H3}.
Our results are summarized as follows, where we ignore negligible polynomial factors of arbitrarily small degrees for simplicity.
\begin{enumerate}
  \item (Eq.~\eqref{eq:bound_lambda_thm}) We show that, with our algorithm, the excess classification error $\epsilon$ \textit{converges at an exponential speed}, i.e., $\epsilon=O(\exp(-N))$ in terms of the number $N$ of examples. This bound exponentially improves the previously known bound $\epsilon=O(\poly(\nicefrac{1}{N}))$ by~\citet{NIPS2018_7598} in cases of using optimized RFs for the same classification task under the low-noise condition.

  \item (Eq.~\eqref{eq:bound_lambda_thm}) We show that our algorithm only requires the number of optimized RFs scaling as $M=O(d(\delta^2)\log(\nicefrac{d(\delta^2)}{\epsilon}))$ for achieving the excess classification error $\epsilon$, where $\delta$ is the constant characterizing the low-noise condition. Here, $d(\cdot)$ is called the \textit{degree of freedom}, defined later in~\eqref{eq:degree_of_freedom}, and shows the \textit{minimum} required number of RFs~\citep{B1}, which our bound achieves up to the logarithmic gap. Significantly, $\epsilon$ \textit{converges at an exponential speed} also in $M$, i.e., $\epsilon=O(\exp(-M))$, achieving as fast scaling in $M$ as that in $N$.

  \item (Eq.~\eqref{eq:T_all}) As a whole, we prove that the runtime of our algorithm is $\widetilde{O}(MD\log(\nicefrac{1}{\epsilon}))$, which is poly-logarithmic in the error $\nicefrac{1}{\epsilon}$ and linear in data dimension $D$. Since we use optimized RFs, the number $M=O(d(\delta^2)\log(\nicefrac{d(\delta^2)}{\epsilon}))$ of RFs can also be significantly smaller than the conventional RFs\@. Thus, our analysis finds a promising application of QML, making it possible to establish a significantly fast algorithm for classification with optimized RFs.
\end{enumerate}

To prove these results, we develop new theoretical techniques that are crucial for the speedup, progressing beyond the previous works by~\citet{NIPS2018_7598,NEURIPS2020_9ddb9dd5,pmlr-v75-pillaud-vivien18a,pmlr-v89-nitanda19a,Yashima2019}.
We will summarize these techniques in Sec.~\ref{sec:runtime}.
Although~\citet{NEURIPS2020_9ddb9dd5} solves the regression problem by the quantum algorithm, how to achieve classification using the quantum algorithm without ruining applicability and the convergence speed of the existing state-of-the-art classification algorithms is far from trivial since the classification is a different task.
After all, the runtime of the regression obtained by~\citet{NEURIPS2020_9ddb9dd5} was
$\tilde{O}(\nicefrac{MD}{\epsilon^2})$ (to accuracy $\epsilon$), but our algorithm achieves the classification task in a different runtime $\tilde{O}(MD\log(\nicefrac{1}{\epsilon}))$
(to the excess classification error $\epsilon$).
The existing analysis by~\citet{NIPS2018_7598} of the same classification task as ours was insufficient to prove this runtime $\tilde{O}(MD\log(\nicefrac{1}{\epsilon}))$ with the optimized RFs; that is, a mere use of the quantum algorithm of~\citet{NEURIPS2020_9ddb9dd5} combined with the analysis by~\citet{NIPS2018_7598} would ruin the exponential convergence speed in the classification under the low-noise condition.
By contrast, our techniques improve the results of~\citet{NIPS2018_7598} exponentially in $\epsilon$, achieving as fast convergence speed as the state-of-the-art algorithms by~\citet{pmlr-v75-pillaud-vivien18a,pmlr-v89-nitanda19a,Yashima2019} in $N$ while significantly saving the required number $M$ of RFs for scalability.
We also remark that~\citet{Yashima2019} has analyzed the use of conventional RFs for the classification under the low-noise condition, but a straightforward application of their technique to our setting using optimized RFs would only lead to the runtime $\tilde{O}(\poly(\nicefrac{1}{\epsilon}))$ in $\epsilon$, ruining our exponentially fast runtime $\tilde{O}(\log(\nicefrac{1}{\epsilon}))$ in $\epsilon$; as we will discuss in Sec.~\ref{sec:runtime}, the techniques developed here are crucial for avoiding this slowdown in bounding the runtime.

\textit{Impact in the field of QML}. ---
The novelty of our results is to accelerate classification, a common task in machine learning, by taking advantage of the exponential speedup in the QML algorithm of~\citet{NEURIPS2020_9ddb9dd5}; significantly, we achieve this acceleration without ruining broad applicability and exponential error convergence of the leading kernel-based algorithms for classification under the low-noise condition by~\citet{pmlr-v75-pillaud-vivien18a,pmlr-v89-nitanda19a,Yashima2019}.
Conventionally, QML algorithms for classification such as quantum support vector machine~\citep{PhysRevLett.113.130503} may achieve large speedups compared to classical algorithms only under restrictive assumptions that matrices involved in the algorithms are sparse or have low rank.
More recent ``quantum-inspired'' classical algorithms for classification~\citep{chen2019,arXiv:1910.06151} also require low rank.
However, the sparsity and low-rankness assumptions limit the power and applicability of the QML algorithms; that is, to take advantage of the large speedups of QML, careful justifications of the sparsity and low-rankness assumptions have been needed~\citep{aaronson2015read}.
By contrast, our algorithm, based on the quantum algorithm of~\citet{NEURIPS2020_9ddb9dd5}, does not require the sparsity or low-rankness assumptions yet can benefit from the exponential quantum speedup in sampling the optimized RFs, leading to the significant reduction of the required number $M$ of RFs as discussed above.

We stress that, despite the efforts to exploit QML for accelerating kernel methods~\citep{Mengoni2019}, it was challenging to use the exponential speedup without ruining applicability due to the sparsity and low-rankness assumptions, except for the regression achieved by~\citet{NEURIPS2020_9ddb9dd5}.
For example, a quantum algorithm for classification by~\citet{pmlr-v97-li19b} may not need the sparsity and low-rankness assumptions, but the speedup compared to classical algorithms is not exponential but polynomial; that is, the classical computation can simulate it up to polynomial-time overhead.
Heuristic QML algorithms for noisy quantum devices such as that of~\citet{havlivcek2019supervised} may not require the sparsity and low-rankness assumptions, but no proof bounds its runtime.
QML may have advantages in learning data obtained from quantum states~\citep{Sweke2021quantumversus,PhysRevLett.126.190505,doi:10.1126/science.abn7293,9719827}, but tasks in the field of machine learning conventionally deal with classical data rather than quantum states; problematically, it is unknown whether QML algorithms in the quantum settings of these works are applicable to acceleration in the common learning tasks for the conventional classical data, such as the classification of classical data.
For a classical data set constructed carefully so that its classification reduces to a variant of Shor's algorithm~\citep{10.1137/S0097539795293172}, QML may achieve the classification super-polynomially faster than classical algorithms~\citep{Yunchao2020}; however, the applicability of such QML to practical data sets has been unknown, unlike the broad applicability of kernel methods implemented by classical computation.

In contrast to these QML algorithms, our algorithm can indeed take advantage of the exponential quantum speedup meaningfully for the acceleration of the classification.
QML using exponential speedup is hard to simulate by classical computation and hard to perform even on near-term noisy small- or intermediate-scale quantum devices.
But we aim at classification at a large scale since we will eventually need large-scale machine learning in practice.
For this reason, we do not use numerical simulation but prove the advantage analytically.
Remarkably, we can obtain the advantage of our algorithm in reducing the number of RFs within the same model as conventional RFs; as a result, the advantage is expected to appear for a practical class of data sets, e.g., those given by a sub-Gaussian or spherical data distribution learned with the Gaussian kernel (see~\eqref{eq:advantage}).
Thus, our results are fundamental for establishing the widely applicable framework of QML that takes advantage of exponential quantum speedups to accelerate various learning tasks, in particular, not only regression but also classification;
moreover, even if large-scale quantum computation is yet to be realized under current technology, our theoretical results have an impact on further developments of quantum technology toward realizing QML, by providing solid motivation with theoretical guarantee.

\section{\label{sec:problem}Setting of Classification with Optimized Random Features}

We consider binary classification with $D$-dimensional input $x\in\cX=\mathbb{R}^D$ and its binary label $y\in\mathcal{Y}=\{-1,+1\}$.
Our learning algorithm takes an approach of semi-supervised learning, where we have many unlabeled input examples and much fewer labeled examples.
Suppose that $N$ pairs of labeled examples $(x_0,y_0),(x_1,y_1),\ldots,(x_{N-1},y_{N-1})\in\cX\times\mathcal{Y}$ are given according to observation of independent and identically distributed (IID) random variables $(X,Y)$ with a probability distribution $\rho(x,y)$.
The marginal probability distribution on $\cX$ for $\rho(x,y)$ is denoted by $\rho_\cX(x)$, and the support of $\rho_\cX$ by $\cX_\rho=\supp(d\rho_\cX)$.
In addition, $N_0$ unlabeled examples $x\in\cX$ ($N_0\gg N$) are given according to $\rho_\cX(x)$.
Given $x\in\cX$, the conditional distribution of $Y$ conditioned on $x$ is denoted by $\rho(y|x)$.

The goal of binary classification is to estimate a function $f:\cX\to\mathbb{R}$ whose sign $\pm$ indicates as correct labels of the input data as possible.
In particular, we want to learn a classifier minimizing the classification error
$\cR(f)\coloneqq\mathbb{E}\left[I(\sign(f(X)),Y)\right]$,
where $\sign(x)=\nicefrac{x}{|x|}$, and $I$ is the $0$-$1$ error function, i.e., $I(y,y^\prime)=0$ if $y=y^\prime$, and $I(y,y^\prime)=1$ if $y\neq y^\prime$.
Let $\cR^\ast$ denote the minimal achievable classification error given by $\cR^\ast\coloneqq\cR(\mathbb{E}[Y|X])$, and $f^\ast$ denote a minimizer of $\cR$, i.e., the optimal classifier, given by~\citep{pmlr-v75-pillaud-vivien18a}
\begin{equation}
  \label{eq:optimal_classifier}
  f^\ast(x)\coloneqq\mathbb{E}[Y|X=x]=\rho(1|x)-\rho(-1|x).
\end{equation}
Given a desired error $\epsilon>0$,
the classification task aims to obtain an estimate $\hat{f}$ of $f^\ast$ with the excess classification error $\cR(\hat{f})-\cR^\ast$ bounded by
\begin{equation}
  \label{eq:excess_testing_error}
  \mathbb{E}[\cR(\hat{f})-\cR^\ast]\leqq\epsilon,
\end{equation}
where the expectation is taken over multiple runs of the learning algorithm since our algorithm will be a randomized algorithm.

However, direct minimization of the classification error  $\cR$ may be intractable due to discontinuity and nonconvexity of $\cR$.
Thus, it is conventional to solve the problem of classification by minimizing a continuous and convex loss function rather than minimizing the classification error directly.
In particular, following~\citet{pmlr-v75-pillaud-vivien18a}, we will use a square loss function, or the loss for short, given by
\begin{equation}
  \label{eq:square_loss}
  \cL(f)\coloneqq\mathbb{E}\Big[{|Y-f(X)|}^2\Big].
\end{equation}

We will learn the classifier based on the kernel methods;
i.e.,  as a model of functions to be learned, we use the reproducing kernel Hilbert space (RKHS) $\mathcal{F}$ associated with the kernel $k$~\citep{S5,W2,Steinwart2008}.
Since our goal is to show a runtime advantage of the optimized RFs in a concrete setting, we here focus on using the Gaussian kernel, i.e., $k\big(x,x^\prime\big)\coloneqq\exp\big(-\gamma\|x-x^\prime\|_2^2\big)$ for $\gamma>0$.
In this case, the RKHS $\mathcal{F}$ of the Gaussian kernel is universal in $L^2(d\rho_\cX)$~\citep{Steinwart2008,JMLR:v7:micchelli06a,S4}; that is, any square-integrable function can be well approximated with our model $\mathcal{F}$.

Our assumptions are as follows.
\begin{enumerate}
  \item As in any implementation of kernel methods by computer, we represent real numbers by a finite number of bits and quantum bits (qubits) in terms of fixed-point or floating-point number representation with sufficiently high precision $\Delta$; by convention, our analysis ignores errors $O(\Delta)$ and poly-logarithmic overheads $O(\polylog(\nicefrac{1}{\Delta}))$ arising from this discretization.
  \item Due to the uniform law of large numbers, we can approximate the true data distribution $\rho_\mathcal{X}$ by the empirical distribution $\hat{\rho}_\mathcal{X}$ of the $N_0$ unlabeled examples with error $O(\nicefrac{1}{N_0^2})$, which we consider to be negligibly small by assuming $N_0\gg N$. In particular, by taking $N_0\approx\nicefrac{1}{\Delta^2}$, we assume $\|\rho_\mathcal{X}-\hat{\rho}_\mathcal{X}\|_1=O(\Delta)$, which is within the error of discretization that we ignore.
\item We assume that $\cX_\rho$ is bounded, and $\rho_\cX$ is uniformly bounded away from $0$ and $\infty$ on $\cX_\rho$.
\item Following the previous work by~\citet{NIPS2018_7598} studying the classification with optimized RFs, we assume a (strong) low-noise condition
\begin{equation}
  \label{eq:low_noise}
  \left|\mathbb{E}[Y|X]\right|>\delta,\quad\text{almost surely},
\end{equation}
where $\delta$ is a fixed parameter with $0<\delta\leqq 1$.
This condition is well studied in the analysis of classification tasks, also known as Massart's low-noise assumption~\citep{koltchinskii2011oracle}.
As the amount of noise gets low, $\delta$ becomes large.
\item We assume that optimal $f^\ast$ on $\cX_\rho$ is a function in our model $\mathcal{F}$, i.e., the RKHS of the Gaussian kernel.
Then, the minimizer of the loss~\eqref{eq:square_loss} in $\mathcal{F}$ is indeed achieved by $f^\ast\in\mathcal{F}$~\citep{pmlr-v75-pillaud-vivien18a}, which our algorithm will estimate by $\hat{f}_{v,\alpha}$ in~\eqref{eq:estimate}.
\end{enumerate}

\section{Main Results}

\begin{algorithm}[t]
  \caption{\label{alg:classification}Classification with optimized RFs.}
  \begin{algorithmic}[1]
    \REQUIRE{Parameter $\lambda$ in~\eqref{eq:bound_lambda_thm}, number $M$ of optimized RFs in~\eqref{eq:bound_lambda_thm}, even number $N\in 2\mathbb{N}$ of labeled examples in~\eqref{eq:bound_lambda_thm}, unlabeled input examples stored in the data structure described in the main text, step sizes $\eta_t$ in~\eqref{eq:eta}, parameter $q_{\min}$ in~\eqref{eq:regularized_loss_alpha}, parameter region $\mathcal{W}\subset\mathbb{R}^{2M}$.}
    \ENSURE{A classifier $\hat{f}_{v,\alpha}$ in the form of~\eqref{eq:estimate} with $v_0,\ldots,v_{M-1}\in\mathcal{V}=\mathbb{R}^D$ and $\alpha_0,\ldots,\alpha_{2M-1}\in\mathbb{R}$ that achieves~\eqref{eq:excess_testing_error}, i.e., the excess classification error bounded by $\epsilon$ in expectation, with exponential error convergence shown in Theorem~\ref{thm:generalization} and within a remarkably short runtime $T_\mathrm{all}$ in~\eqref{eq:T_all}.}
    \FOR{$m\in\left\{0,\ldots,M-1\right\}$}
    \STATE{Sample an optimized RF $v_m\in\mathcal{V}$ according to the optimized distribution $q_\lambda^\ast(v)d\tau(v)$ in~\eqref{eq:q_complex}.}
    \COMMENT{See \textit{Sampling optimized RFs} in Sec.~\ref{sec:algorithm} for detail.}
    \ENDFOR%
    \STATE{Initialize $\alpha^{(0)}\gets{(0,0,\ldots,0)}\in\mathcal{W}$.}
    \COMMENT{See \textit{SGD} in Sec.~\ref{sec:algorithm} for detail.}
    \FOR{$t\in\left\{0,\ldots,N-1\right\}$}
    \STATE{Calculate the prefactor $C(\alpha^{(t)})$ of the unbiased estimate $\hat{g}^{\left(t\right)}$ of the gradient of the regularized testing loss $\cL_\lambda$ in~\eqref{eq:regularized_loss_alpha} at $\alpha^{(t)}$ according to~\eqref{eq:unbiased_estimate}.}
    \STATE{Calculate $\hat{g}^{\left(t\right)}$ according to~\eqref{eq:unbiased_estimate} using $C(\alpha^{(t)})$.}
    \STATE{Set $\alpha^{\left(t+1\right)}\gets\Pi_\mathcal{W}(\alpha^{\left(t\right)}-\eta^{\left(t\right)} \hat{g}^{\left(t\right)})$.}
    \COMMENT{Using projection $\Pi_\mathcal{W}$ onto $\mathcal{W}$.}
    \ENDFOR%
    \STATE{Set ${(\alpha_0,\ldots,\alpha_{2M-1})}\gets(\nicefrac{2}{N})\sum_{t=\nicefrac{N}{2}+1}^{N}\alpha^{(t)}$.}
    \COMMENT{Suffix averaging.}
    \STATE{\textbf{Return } $v_0,\ldots,v_{M-1}\in\mathcal{V}$ and $\alpha_0,\ldots,\alpha_{2M-1}\in\mathbb{R}$.}
    \COMMENT{i.e., $\hat{f}_{v,\alpha}$ with these parameters.}
  \end{algorithmic}
\end{algorithm}

Our main contribution is to develop Algorithm~\ref{alg:classification} using optimized RFs and SGD to solve the classification problem formalized in Sec.~\ref{sec:problem}, and prove its remarkably short runtime.
In Sec.~\ref{sec:algorithm}, we describe the algorithm.
In Sec.~\ref{sec:runtime}, we analyze the generalization property and the runtime.

\subsection{\label{sec:algorithm}Algorithm}

In this subsection, we describe our algorithm for achieving the classification task formulated in Sec.~\ref{sec:problem}, which we show in Algorithm~\ref{alg:classification}.
In this algorithm, we sample optimized RFs based on a quantum algorithm of~\citet{NEURIPS2020_9ddb9dd5} and then perform regression using SGD to minimize a regularized version of the loss.

\textit{Sampling optimized RFs}. ---
Algorithm~\ref{alg:classification} starts with performing sampling of $M$ optimized RFs using a quantum algorithm of~\citet{NEURIPS2020_9ddb9dd5}.
The optimized distribution is given by~\citep{B1}
\begin{equation}
  \label{eq:q_complex}
  q_\lambda^\ast(v)d\tau(v)\propto\braket{\varphi(v,\cdot)|{(\Sigma+\lambda\mathbbm{1})}^{-1}\varphi(v,\cdot)}_{L_2(d\rho_\mathcal{X})}d\tau(v),
\end{equation}
where $\lambda>0$ is a hyperparameter used for regularization, and $\Sigma$ is the \textit{integral operator} $(\Sigma f)(x)=\int_\mathcal{X}d\rho_\mathcal{X}(x^\prime)k(x,x^\prime)f(x^\prime)$ on a space $L^2(d\rho_\mathcal{X})$ of functions, depending both on kernel $k$ and data distribution $\rho_\mathcal{X}$.
The condition on $\lambda$ will be clarified later in~\eqref{eq:bound_lambda_thm}.
By repeating the sampling from $q_\lambda^\ast(v)d\tau(v)$ using the quantum algorithm $M$ times, we obtain $M$ optimized RFs $v_0,\ldots,v_{M-1}$, where the condition on $M$ will be clarified later by~\eqref{eq:bound_lambda_thm}.

Although~\citet{NEURIPS2020_9ddb9dd5} may have developed core methods for this sampling,
the quantum algorithm of~\citet{NEURIPS2020_9ddb9dd5} has several parameters that need to be tuned depending on learning tasks to be solved;
recalling that it is usually hard to attain widely applicable speedup in QML as summarized in Sec.~\ref{sec:intro},
we need to prove that we can choose these parameters appropriately, and we can fulfill the requirement for running this quantum algorithm even in our setting of classification, so as to clarify that the sampling is feasible for our classification task.
The quantum algorithm of~\citet{NEURIPS2020_9ddb9dd5} uses a technique for decomposing representation of $q_\lambda^\ast(v)d\tau(v)$ on the right-hand side of~\eqref{eq:q_complex} into operators that quantum computation can efficiently deal with, and combines this decomposition with two techniques in quantum computation, quantum singular value transformation~\citep{G1} and quantum Fourier transform~\citep{C5,H2}, to achieve exponential speedup in sampling optimized RFs compared to the existing classical sampling algorithms by~\citet{B1,NIPS2018_7598,Shahrampour2019}.
As discussed in Sec.~\ref{sec:intro}, the significance of this quantum algorithm is to avoid restrictive assumptions on sparsity and low-rankness of the operators owing to this decomposition; as a result, the algorithm is widely applicable to representative choices of kernels, including the Gaussian kernel used here~\citep{NEURIPS2020_9ddb9dd5}.
The true data distribution $\rho_\mathcal{X}$ to compute $q_\lambda^\ast(v)$ may be unknown, but we here design our setting to be semi-supervised learning, so that we can use the empirical data distribution $\hat{\rho}_\mathcal{X}$ of unlabeled data as a good approximation of $\rho_\mathcal{X}$ without the cost of labeling all data.
The input model for the quantum algorithm is to prepare $\sum_x\sqrt{\hat{\rho}_\mathcal{X}(x)}\Ket{x}$, i.e., a quantum state that can produce a random bit string sampled from $\hat{\rho}_\mathcal{X}$ as a measurement outcome if measured in the standard basis.
For classical algorithms, sampling from $\hat{\rho}_\mathcal{X}$ can be realized easily in poly-logarithmic runtime $O(D\polylog(N_0))$ in $N_0$, by sampling $n_0\in\{1,\ldots,N_0\}$ from the uniform distribution over $\log_2 (N_0)$ bits and outputting $x_{n_0}$ out of $x_1,\ldots,x_{N_0}$ stored in random access memory (RAM).
As for the quantum algorithm, the preparation of $\sum_x\sqrt{\hat{\rho}_\mathcal{X}(x)}\Ket{x}$ with maintaining quantum superposition may be more technical.
But we show that this preparation is also implementable in runtime $O(D\polylog(N_0))$, by storing the $N_0$ unlabeled examples upon collecting them in a sparse binary tree data structure~\citep{kerenidis_et_al:LIPIcs:2017:8154} with quantum RAM (QRAM)~\citep{PhysRevA.78.052310,PhysRevLett.100.160501}.
See Appendix~F for details.
In our setting of $N_0\approx\nicefrac{1}{\Delta^2}$ with precision $\Delta$, the runtime per inputting $\sum_x\sqrt{\hat{\rho}_\mathcal{X}(x)}\Ket{x}$ is $O(D\polylog(N_0))=O(D\polylog(\nicefrac{1}{\Delta}))$, i.e., the linear scaling in $D$ up to the same overhead as that of discretization, making the quantum algorithm implementable in a feasible runtime.

\textit{SGD}. ---
After sampling $M$ optimized RFs according to $q_\lambda^\ast(v)d\tau(v)$,
Algorithm~\ref{alg:classification} proceeds to optimize coefficients $\alpha={(\alpha_0,\ldots,\alpha_{2M-1})}\in\mathbb{R}^{2M}$ to learn a classifier $\hat{f}_{v,\alpha}$ in~\eqref{eq:estimate}.
For this optimization, we use SGD with suffix averaging, minimizing a regularized version $\cL_\lambda$ of the loss $\cL(\hat{f}_{v,\alpha})$ in~\eqref{eq:square_loss}.
In particular, we write the loss for $\hat{f}_{v,\alpha}$ in~\eqref{eq:estimate} as
$\cL(\alpha)=\cL(\hat{f}_{v,\alpha})\coloneqq
\mathbb{E}\Big[\Big|Y-\sum_{m=0}^{M-1}\big(\alpha_{2m}\cos(-2\pi v_m\cdot X)+
\alpha_{2m+1}\sin(-2\pi v_m\cdot X)\big)\Big|^2\Big]$,
and define the regularized loss as
\begin{equation}
  \label{eq:regularized_loss_alpha}
  \cL_\lambda(\alpha)=\cL_\lambda(\hat{f}_{v,\alpha})\coloneqq\cL(\alpha)+\lambda Mq_{\min}\|\alpha\|_2^2,
\end{equation}
where $\lambda>0$ is the same parameter as that for $q_{\lambda}^\ast$ in~\eqref{eq:q_complex}, and $q_{\min}\coloneqq\min\{q_\lambda^\ast(v_m):m=0,\ldots,M-1\}$.
We use $q_{\min}$ as a hyperparameter to be chosen appropriately prior to the learning.
In the parameter region of sampling optimized RFs that are weighted by importance and that nearly minimize $M$, the minimal weight $q_{\min}$ of the sampled RFs can be considered to be a constant bounded away from $0$.
To guarantee this more explicitly, we can also use a bottom-raised weight in place of $q_\lambda^\ast(v_m)$, e.g., $\nicefrac{q_\lambda^\ast(v_m)}{2}+\nicefrac{1}{2}$;
in this case, each optimized RF $v_m$ is to be sampled with weight $q_\lambda^\ast(v_m)$ once in two samplings in expectation, and hence at most twice as many RFs as those with $q_\lambda^\ast(v_m)$ suffice to achieve the learning, while we can consider $q_{\min}$ to be $\nicefrac{1}{2}$, a constant.
The term $\lambda Mq_{\min}\|\alpha\|_2^2$ makes the regularized loss $\cL_\lambda$ $\mu$-strongly convex for a constant $\mu=\lambda M q_{\min}$, leading to accelerating SGD compared to minimizing $\cL$ without regularization.
The SGD looks for optimal $\alpha\in\mathcal{W}\subset\mathbb{R}^{2M}$ by minimizing $\mathcal{L}_\lambda(\alpha)$ within a parameter region $\mathcal{W}$,
which is chosen as a $2M$-dimensional ball centered at the origin and of a diameter shown later in~\eqref{eq:alpha}.

To guarantee the runtime of SGD theoretically, we clarify the setting of SGD as follows.
In the SGD, we start from an initial point $\alpha^{(0)}\in\mathcal{W}$;
following~\citet{H3}, we here choose $\alpha^{(0)}$ as the origin.
Using $N$ labeled examples, we iteratively update the points $\alpha^{(1)},\ldots,\alpha^{(N)}\in\mathcal{W}$ in total $N$ times, where for each $t\in\{0,\ldots,N\}$, we write
$\alpha^{(t)}=\big(\alpha^{(t)}_0,\ldots,\alpha^{(t)}_{2M-1}\big)\in\mathcal{W}$.
The required number $N$ of labeled examples will be clarified later in~\eqref{eq:bound_lambda_thm}.
We update the point from $\alpha^{(t)}$ to $\alpha^{(t+1)}$ in the $(t+1)$th iteration of SGD using an unbiased estimate $\hat{g}^{(t)}$ of the gradient of $\cL_\lambda$ at $\alpha^{(t)}$, i.e.,
$\mathbb{E}\big[\hat{g}^{(t)}\big]=\nabla\cL_\lambda(\alpha^{(t)})$.
To obtain $\hat{g}^{(t)}$,
we regard the sequence of the $N$ labeled examples $\left(x_0,y_0\right),\ldots,\left(x_{N-1},y_{N-1}\right)$ as a data stream that sequentially provides the examples as observation of IID random variables according to $\rho(x,y)$.
Using these examples, we construct $\hat{g}^{(t)}$ by
\begin{align}
  \label{eq:unbiased_estimate}
\hat{g}^{(t)}\coloneqq C\big(\alpha^{(t)}\big)\left(\begin{matrix}
      \cos(-2\pi v_0\cdot x_t)\\
      \sin(-2\pi v_0\cdot x_t)\\
      \vdots\\
      \cos(-2\pi v_{M-1}\cdot x_t)\\
      \sin(-2\pi v_{M-1}\cdot x_t)
      \end{matrix}\right)+2\lambda Mq_{\min}\left(\begin{matrix}
      \alpha^{(t)}_0\\
      \alpha^{(t)}_1\\
      \vdots\\
      \alpha^{(t)}_{2M-2}\\
      \alpha^{(t)}_{2M-1}
  \end{matrix}\right),
\end{align}
where $C\big(\alpha^{(t)}\big)$ is a prefactor shared among all the $2M$ elements of $\hat{g}^{(t)}$, i.e.,
$C\big(\alpha^{(t)}\big)\coloneqq 2\big(y_t-\sum_{m=0}^{M-1}(\alpha_{2m}^{(t)}\cos(-2\pi v_m\cdot x_t)+\alpha_{2m+1}^{(t)}\sin(-2\pi v_m\cdot x_t))\big)$.
To save runtime in calculating $\hat{g}^{(t)}$, Algorithm~\ref{alg:classification} calculates and stores the shared prefactor $C\big(\alpha^{(t)}\big)$ only once per iteration, and then calculates $\hat{g}^{(t)}$ by multiplying $C^{(t)}$ and $\cos(-2\pi v_0\cdot x_t),\ldots,\sin(-2\pi v_{M-1}\cdot x_t)$ to obtain each of the $2M$ elements on the right-hand side of~\eqref{eq:unbiased_estimate}.
With this $\hat{g}^{(t)}$, SGD would calculate $\alpha^{(t+1)}\in\mathcal{W}$ based on $\alpha^{(t)}-\eta^{(t)}\hat{g}^{(t)}$ using a hyperparameter  $\eta^{(t)}$ representing step size,
which we take as a decaying sequence used by~\citet{H3}, i.e.,
\begin{equation}
  \label{eq:eta}
  \eta^{(t)}=O\left(\nicefrac{1}{\mu t}\right).
\end{equation}
However, if we update $\alpha^{(t)}$ to $\alpha^{(t+1)}$ in this way, $\alpha^{(t+1)}$ may be out of $\mathcal{W}$, which may slow down the SGD potentially.
To avoid this slowdown, the update for each iteration of SGD here uses a projection onto $\mathcal{W}$; that is, $\alpha^{(t+1)}$ is obtained from $\alpha^{(t)}$ by
$\alpha^{(t+1)}=\Pi_\mathcal{W}\left(\alpha^{(t)}-\eta^{(t)}\hat{g}^{(t)}\right)\in\mathcal{W}$,
where $\Pi_\mathcal{W}$ denotes the projection onto $\mathcal{W}$.
Then, the SGD provably converges to the minimizer of the regularized loss $\cL_\lambda$~\citep{H3}.
All the parameters of Algorithm~\ref{alg:classification} are chosen carefully so that this minimizer achieves our classification task; that is, $\hat{f}_{v,\alpha}$ obtained from Algorithm~\ref{alg:classification} achieves~\eqref{eq:excess_testing_error}.

\subsection{\label{sec:runtime}Generalization Property and Runtime}

In this subsection, we show the generalization property of optimized RFs in Algorithm~\ref{alg:classification} for the classification task formulated in Sec.~\ref{sec:problem}, and clarify the advantage of Algorithm~\ref{alg:classification} in terms of runtime.
The minimal required number $M$ of RFs, which is achievable by the optimized RFs, is characterized by the degree of freedom $d(\lambda)$ for an appropriate parameter $\lambda$ in such a way that~\citep{B1}
\begin{equation}
\label{eq:degree_of_freedom}
M=\widetilde{O}(d(\lambda)),\quad d(\lambda)\coloneqq\tr\Sigma{(\Sigma+\lambda\mathbbm{1})}^{-1},
\end{equation}
where $\widetilde{O}$ ignores poly-logarithmic factors, and $\Sigma$ is the integral operator in~\eqref{eq:q_complex}.
For our classification task,
we prove the following theorem on the required number of RFs and labeled examples.
To simplify the presentation, we ignore arbitrarily small-degree polynomial factors in our bounds;
see also the beginning of Appendices for the formal statement of the theorem without ignoring these factors.
Significantly, owing to the low-noise condition~\eqref{eq:low_noise}, we achieve exponentially fast convergence of the excess classification error in $N$ even with optimized RFs, i.e., $\epsilon=O(\exp(-N))$; moreover, we show exponential error convergence also in $M$, i.e., $\epsilon=O(\exp(-M))$, with the optimized RFs to achieve the minimum $M$.
The proof techniques that we develop will be summarized later, after discussing the advantage.

\begin{theorem}[\label{thm:generalization}Informal. Generalization property of optimized RFs in Algorithm~\ref{alg:classification}]
  There exist the parameter $\lambda$, the number $M$ of optimized RFs, and the number $N$ of labeled examples satisfying, up to ignoring polynomial factors of an arbitrarily small degree,
  \begin{align}
    \label{eq:bound_lambda_thm}
    \lambda&=O(\nicefrac{\delta^{2}}{\|f^\ast\|_\mathcal{F}^2}),\quad M=O(d(\lambda)\log(\nicefrac{d(\lambda)}{\epsilon})),\quad N=O\Big(\log(\nicefrac{1}{\epsilon})(\nicefrac{\|f^\ast\|_\mathcal{F}^4}{(\delta^{4} q_{\min}^2)})\Big),
  \end{align}
  such that Algorithm~\ref{alg:classification} can return a classifier $\hat{f}_{v,\alpha}$ satisfying the $\epsilon$-small excess classification error~\eqref{eq:excess_testing_error},
  where $d(\lambda)$ is defined as~\eqref{eq:degree_of_freedom}, $\|f^\ast\|_\mathcal{F}$ is the RKHS norm of the optimal classifier $f^\ast\in\mathcal{F}$ in~\eqref{eq:optimal_classifier}.
\end{theorem}

\textit{Advantage in terms of runtime}. ---
The novelty of our classification algorithm, Algorithm~\ref{alg:classification}, is to use the optimized RFs sampled by the quantum algorithm within a feasible runtime, which makes it possible to achieve the classification within a remarkably short runtime with the minimal number $M$ of RFs.
For representative choices of kernels such as the Gaussian kernel, the quantum algorithm can sample optimized RFs according to $q_\lambda^\ast(v)d\tau(v)$ within runtime per sampling as fast as $T_\mathrm{sampling}=\widetilde{O}\left(\nicefrac{D}{\lambda}\right)$~\citep{NEURIPS2020_9ddb9dd5}.
We also bound the runtime per iteration of SGD in Algorithm~\ref{alg:classification} as follows.
The prefactor $C^{(t)}$ in~\eqref{eq:unbiased_estimate} consists of the sum of $2M$ terms with each term including an inner product of $D$-dimensional vectors, requiring $O(MD)$ runtime to calculate.
Given $C^{(t)}$, the unbiased estimate $\hat{g}^{(t)}$ in~\eqref{eq:unbiased_estimate} of the gradient consists of $2M$ elements with each element including an inner product of $D$-dimensional vectors, requiring $O(MD)$ runtime to calculate.
Given $\hat{g}^{(t)}$, the update from $\alpha^{(t)}$ to $\alpha^{(t+1)}$ is arithmetics of $2M$-dimensional vectors, requiring $O(M)$ runtime.
Thus, the runtime per iteration is $T_\mathrm{iterate}=O(MD)$.
Consequently, Theorem~\ref{thm:generalization} shows that the overall runtime $T_\mathrm{all}$ of Algorithm~\ref{alg:classification} is, in terms of $M,D,\epsilon$,
\vspace{-0.1cm}
\begin{align}
  \label{eq:T_all}
  T_\mathrm{all}&=MT_\mathrm{sampling}+NT_\mathrm{iterate}=\widetilde{O}(\nicefrac{MD}{\lambda}+NMD)=\widetilde{O}(MD\log(\nicefrac{1}{\epsilon})),
\end{align}
where $M=O(d(\lambda)\log(\nicefrac{d(\lambda)}{\epsilon}))$;
hence, $T_\mathrm{all}$ is as fast as poly-logarithmic in $\nicefrac{1}{\epsilon}$ and linear in $D$.

Remarkably, the use of optimized RFs sampled from $q_\lambda^\ast(v)d\tau(v)$ can significantly reduce the required number $M$ of RFs in~\eqref{eq:T_all}.
How large $M=\widetilde{O}(d(\lambda))$ should be is determined by the decay of eigenvalues ${(\mu_i)}_{i\in\mathbb{N}}$ of the integral operator $\Sigma$ in~\eqref{eq:q_complex}~\citep{B1}.
In the worst case $\mu_i=O(\nicefrac{1}{i})$, we would need as large as $M=\widetilde{O}(d(\lambda))=\widetilde{O}(\nicefrac{1}{\lambda})=\widetilde{O}(\nicefrac{1}{\delta^2})$ in terms of $\delta$ for the low-noise condition~\eqref{eq:low_noise}~\citep{B1}, which is the same scaling as conventional RFs for approximating the square loss function~\citep{R2,R3,NIPS2017_6914}.
\citet{Yashima2019} has also studied an algorithm with the conventional RFs and SGD for the classification task under the low-noise condition, but the required number of RFs has been as large as $M=O(\poly(\nicefrac{1}{\delta}))$ in the same way.
By contrast, for the Gaussian kernel and a sub-Gaussian data distribution, the decay is $\mu_i=O(\exp(-i^{\nicefrac{1}{D}}))$~\citep{NIPS2018_7598,NIPS2007_4ca82782}, and we indeed have an exponential advantage
$M=\widetilde{O}(d(\lambda))=\widetilde{O}(\log^D(\nicefrac{1}{\lambda}))=\widetilde{O}(\log^D(\nicefrac{1}{\delta^2}))$ in $\delta$.
Moreover, for the Gaussian kernel and a uniform data distribution supported on a sphere $\cX_\rho=S^{D-1}\subset\cX=\mathbb{R}^D$,
the decay is $\mu_i=O(\exp(-i) {(\nicefrac{1}{i})}^{i+(\nicefrac{(D-1)}{2})})=O(\exp(-i))$~\citep{AZEVEDO201457,10.1007/11776420_14}, and we have an exponential advantage in $\delta$ that is independent of $D$, i.e.,
\vspace{-0.1cm}
\begin{equation}
  \label{eq:advantage}
  M=\widetilde{O}(d(\lambda))=\widetilde{O}(\log(\nicefrac{1}{\lambda}))=\widetilde{O}(\log(\nicefrac{1}{\delta})).
\end{equation}
Thus, our results with the optimized RFs can be provably advantageous: we can exponentially reduce the required number $M$ of RFs in $\delta$ compared to the conventional RFs, yet without canceling out the exponential error convergence in $N$ shown by~\citet{pmlr-v75-pillaud-vivien18a,pmlr-v89-nitanda19a,Yashima2019}.
We remark that the runtime $T_\mathrm{all}$ in~\eqref{eq:T_all} may still require a polynomial time in $\nicefrac{1}{\lambda}$ and hence in $\nicefrac{1}{\delta}$, but significantly,  once the classifier $\hat{f}_{v,\alpha}$ is learned by Algorithm~\ref{alg:classification}, the required runtime for each evaluation of $\hat{f}_{v,\alpha}$ is as fast as $O(MD)=\widetilde{O}(D\log(\nicefrac{1}{\delta}))$ in the case of~\eqref{eq:advantage}, i.e., exponentially faster in $\delta$ compared to that of conventional RFs with runtime $O(MD)=\widetilde{O}(\nicefrac{D}{\delta^2})$.
This runtime advantage in using the learned classifier $\hat{f}_{v,\alpha}$ is considerable, especially for applications that require real-time computing, e.g., an embedded system, robotics, feedback control in physical experiments, and machine-learning-assisted quantum error correction.

Finally, we remark that, to take this advantage, it is essential to minimize $M$ by sampling from $q_\lambda^\ast(v)d\tau(v)$.
A difficulty in sampling from $q_\lambda^\ast(v)d\tau(v)$ is that $\Sigma$ in~\eqref{eq:q_complex} is infinite-dimensional.
Even if we discretize $\Sigma$ as done by~\citet{NEURIPS2020_9ddb9dd5}, $\Sigma$ becomes an $O(\exp(D))$-dimensional operator, and the matrix inversion in calculating~\eqref{eq:q_complex} would require $O(\exp(D))$ runtime as long as we use existing classical algorithms.
Another classical algorithm by~\citet{B1} estimates the value of  $q_\lambda^\ast(v)$ based on conventional RFs; however,~\citet{B1} argues its hardness in practice, and its runtime in $D$ is unknown in general.
More problematically, even if we could manage to estimate $q_\lambda^\ast(v)$, the sampling would still be computationally hard due to high dimension.
After all, heuristic sampling algorithms such as Markov chain Monte Carlo methods do not provide theoretical guarantee, and methods with theoretical guarantee based on rejection sampling may require $O(\exp(D))$ runtime per sampling in the worst case since $q_\lambda^\ast(v)$ can be exponentially small in $D$.
Note that~\citet{pmlr-v70-avron17a,Liu2019,NEURIPS2020_012d9fe1,pmlr-v97-li19k} also propose to sample RFs from weighted distribution similar to $q_\lambda^\ast(v)$ in polynomial time by classical algorithms;
however, as discussed by~\citet{NEURIPS2020_9ddb9dd5}, sampling from these similar distributions does not necessarily minimize $M$, i.e., does not lead to our theoretically guaranteed advantage~\eqref{eq:advantage}, since the approximations are heuristic unlike that of~\citet{B1,NIPS2018_7598,Shahrampour2019,NEURIPS2020_9ddb9dd5}.
Similarly, an importance-weighted distribution may also be used in low-rank matrix approximation, i.e., column sampling, but algorithms in the setting of the column sampling~\citep{pmlr-v30-Bach13,NIPS2015_5716,NIPS2018_7810} do not apply to RFs~\citep{B1}.
Quasi-Monte Carlo techniques~\citep{10.5555/2946645.3007073,10.5555/3172077.3172095} can also improve $M$, but it is unknown whether they can achieve minimal $M$.
By contrast, our algorithm minimizes $M$ in feasible runtime.

\textit{Proof techniques}. ---
To prove Theorem~\ref{thm:generalization}, we develop the following techniques.
See Appendices for details.

\begin{itemize}
  \item A technical difficulty in analyzing RFs is that, although the optimal classifier $f^\ast\in\mathcal{F}$ is in the RKHS $\mathcal{F}$ associated with the kernel $k$, the classifier $\hat{f}_{v,\alpha}\in\mathcal{F}_M$ obtained from Algorithm~\ref{alg:classification} is in a different RKHS $\mathcal{F}_M$ associated with an approximated kernel $k_M(x,x^\prime)=\sum_{m=0}^{M-1}(\nicefrac{1}{(Mq_\lambda^\ast(v_m))})\overline{\varphi(v_m,x)}\varphi(v_m,x^\prime)$ with RFs $\varphi(v,x)=\mathrm{e}^{-2\pi\mathrm{i}v\cdot x}$~\citep{B1}.
    Earlier studies of exponential error convergence $\epsilon=O(\exp(-N))$ in the kernel-based classification under the low-noise condition by~\citet{pmlr-v75-pillaud-vivien18a,pmlr-v89-nitanda19a} were based on bounding the RKHS norm in $\mathcal{F}$, but these studies do not apply to the optimized RFs since $\hat{f}_{v,\alpha}$ may not be in $\mathcal{F}$.
    Instead, we use the $L^\infty$ norm. In particular, we show under the low-noise condition~\eqref{eq:low_noise} that, if we have with high probability greater than $1-\epsilon$
    \begin{equation}
      \label{eq:distance}
      \|\hat{f}_{v,\alpha}-f^\ast\|_{L^\infty(d\rho_\cX)}<\delta,
    \end{equation}
    then we achieve the goal of our task, i.e.,~\eqref{eq:excess_testing_error}.
    See Appendix~A\@.
    Note that~\citet{Yashima2019} may have also used the $L^\infty$ norm for analyzing the use of conventional RFs for classification,
    but our analysis for showing the advantage of optimized RFs requires different techniques from~\citet{Yashima2019}, as discussed below.

  \item To bound the distance between $\hat{f}_{v,\alpha}$ and $f^\ast$ in~\eqref{eq:distance}, we need to clarify the required number $M$ of optimized RFs for approximating $f^\ast$ by $\hat{f}_{v,\alpha}$, but at the same time, we need to control $\|\alpha\|_2^2$ of the regularized loss $\cL_\lambda$ in~\eqref{eq:regularized_loss_alpha}. As for $M$,~\citet{B1} has shown that for any $f^\ast\in\mathcal{F}$, with $M$ satisfying $M=\widetilde{O}(d(\lambda))$, we can approximate $f^\ast$ by $\hat{f}_{v,\alpha}\in\mathcal{F}_M$ within $O(\lambda)$ error in the $L^2$ norm, i.e., $\min_{\hat{f}_{v,\alpha}}\|\hat{f}_{v,\alpha}-f\|_{L^2(d\rho_\cX)}^2=O(\lambda)$. But in the work of~\citet{B1}, problematically, $\alpha$ was given depending on values of $q_\lambda^\ast(v_0),\ldots,q_\lambda^\ast(v_{M-1})$, and the regularization for $\alpha$ also needed to use all these values.
    This dependency makes the regularization infeasible; after all, each sampling of a single RF $v$ from $q_\lambda^\ast(v)d\tau(v)$ by itself does not provide the value of $q_\lambda^\ast(v)$.
    In Algorithm~\ref{alg:classification}, we only sample from $q_\lambda^\ast(v)d\tau(v)$ without estimating the value of $q_\lambda^\ast(v)$ since such estimation would cancel out the speedup.
    To address this problem, we here prove that this approximation within $O(\lambda)$ error in the $L^2$ norm is possible with $\alpha$ satisfying
    \begin{equation}
      \label{eq:alpha}
      \|\alpha\|_2 \leqq 2\sqrt{2}(\nicefrac{\|f^\ast\|_\mathcal{F}}{\sqrt{Mq_{\min}}}),
    \end{equation}
    where, importantly, the right-hand side is given in terms of a hyperparameter $q_{\min}$ rather than the values of $q_\lambda^\ast(v_0),\ldots,q_\lambda^\ast(v_{M-1})$.
    See Appendix~B\@.
    Owing to the bound~\eqref{eq:alpha}, we can determine the parameter region $\mathcal{W}$ in Algorithm~\ref{alg:classification} by this hyperparameter $q_{\min}$ without estimating $q_\lambda^\ast(v_0),\ldots,q_\lambda^\ast(v_{M-1})$.

  \item
  Since the above bound for approximating $f^\ast$ by $\hat{f}_{v,\alpha}$ is given in terms of the $L^2$ norm, we furthermore need to develop a technique for translating the $L^2$ norm into a bound in terms of the $L^\infty$ norm used in~\eqref{eq:distance}, in such a way that the exponential convergence $O(\log(\nicefrac{1}{\epsilon}))$ in~\eqref{eq:bound_lambda_thm} should not be canceled out.
    The crucial observation here is that such translation may be possible for the Gaussian kernel by continuously embedding its RKHS in the Sobolev space~\citep{Steinwart2008,Yashima2019,steinwart2009optimal}.
    However, as discussed above, the difficulty is that the embedding of the RKHS $\mathcal{F}$ associated with the Gaussian kernel itself is insufficient since $\hat{f}_{v,\alpha}\in \mathcal{F}_M$ may not be in $\mathcal{F}$.
    For conventional RFs,~\citet{Yashima2019} has shown, using a concentration inequality, i.e., Markov's inequality~\citep{vershynin_2018}, that an embedding of $\mathcal{F}_M$ in the Sobolev space also holds with high probability greater than $1-\epsilon$, but a multiplicative factor $(1+\nicefrac{1}{\epsilon})$ appears in the resulting bound on the $L^\infty$ norm.
    This factor is polynomially large in $\nicefrac{1}{\epsilon}$.
    Problematically, this is insufficient in our setting since the polynomially large factor would cancel out the exponential convergence $O(\log(\nicefrac{1}{\epsilon}))$ in~\eqref{eq:bound_lambda_thm}.
    To address this problem, we here show, using another concentration inequality, i.e., Hoeffding's inequality~\citep{vershynin_2018}, that with high probability greater than $\epsilon$, we have
    $\|\hat{f}_{v,\alpha}-f^\ast\|_{L^\infty(d\rho_\cX)}=O\big(\big(1+\sqrt{\nicefrac{\log(\nicefrac{1}{\epsilon})}{M}}\big)^{\nicefrac{p}{2}}
    {(\|\hat{f}_{v,\alpha}\|_{\mathcal{F}_M}\!\!+\|f^\ast\|_{\mathcal{F}})}^{p}\|f\|_{L^2(d\rho_\cX)}^{1-p}\big)$ for arbitrarily small but fixed $p>0$, where $\|\hat{f}_{v,\alpha}\|_{\mathcal{F}_M}$ is the RKHS norm of $\hat{f}_{v,\alpha}$ in $\mathcal{F}_M$.
    See Apppendix~C\@.
    To prove the exponential convergence,
    the poly-logarithmic prefactor $\big(1+\sqrt{\nicefrac{\log(\nicefrac{1}{\epsilon})}{M}}\big)^{\nicefrac{p}{2}}$ in $\nicefrac{1}{\epsilon}$ is essential.

  \item In addition to the above techniques for bounding $M$, we bound the required number $N$ of labeled examples for minimizing the regularized loss $\cL_\lambda$ using the high-probability bound on the number of iterations in SGD~\citep{H3}.
    Our contribution here is to derive a bound in terms of the $L^\infty$ norm in~\eqref{eq:distance}, by combining the above techniques with the technique for SGD\@.
    See Appendix~D\@.
    In particular, our bound on $\|\alpha\|_2$ in~\eqref{eq:alpha} is crucial for the analysis of SGD based on the technique of~\citet{H3} since~\citet{H3} requires projection onto the parameter region $\mathcal{W}$ in each iteration. Moreover, since $\cL_\lambda$ to be minimized is in terms of the $L^2$ norm, the above translation into the $L^\infty$ norm is vital.
\end{itemize}

\section{Conclusion}

We have constructed a significantly fast algorithm (Algorithm~\ref{alg:classification}) for a classification task under a low-noise condition, using optimized random features (RFs) introduced by~\citet{B1}.
In our algorithm, we sample optimized RFs from a data-optimized distribution by a quantum algorithm of~\citet{NEURIPS2020_9ddb9dd5} within a feasible runtime, followed by stochastic gradient descent (SGD) to achieve exponentially fast convergence of the excess classification error $\epsilon=O(\exp(-N))$ for $N$ examples.
For $D$-dimensional data, our algorithm can achieve the error $\epsilon$ within runtime $\widetilde{O}(MD\log(\nicefrac{1}{\epsilon}))$, where $M=O(d\log(\nicefrac{d}{\epsilon}))$ is the required number of RFs determined by the degree of freedom $d$.
This runtime is as fast as poly-logarithmic in $\nicefrac{1}{\epsilon}$ and linear in $D$;
furthermore, the bound shows $\epsilon=O(\exp(-M))$ in $M$, and hence, our algorithm achieves the exponentially fast error convergence in both $N$ and $M$ simultaneously.
Advantageously, this required number $M$ of optimized RFs can be significantly smaller than conventional RFs that are sampled from a data-independent distribution as originally proposed by~\citet{R2,R3}, and is provably optimal up to a logarithmic gap.
Even more remarkably, we can exploit this advantage for the same classification task as conventional RFs, and thus for a practical class of data sets, e.g., those given by a sub-Gaussian or spherical data distribution to be learned with the Gaussian kernel.
These results discover a promising application of quantum machine learning (QML) to acceleration of leading kernel-based classification algorithms without ruining the broad applicability and the exponential error-convergence speed.

\acks{We would like to acknowledge Taiji Suzuki and Atsushi Nitanda for insightful discussion, and Tam\'{a}s Kriv\'{a}chy for helpful comments on the manuscript. This work was supported by JSPS Overseas Research Fellowships, JST PRESTO Grant Number JPMJPR201A, and JSPS KAKENHI 18K18113\@.}

\newpage

\appendix

\section*{Appendices}

In Appendices of the paper ``Exponential Error Convergence in Data Classification with Optimized Random Features: Acceleration by Quantum Machine Learning'',
we present the proof of the main theorem (Theorem~1 in the main text) on the generalization property of optimized random features (RFs) in our algorithm (Algorithm~1 in the main text), by providing the proof techniques summarized in the main text.
Appendices are organized as follows.
In Appendix~\ref{sec:reduction_to_l_infty_norm}, we show how to reduce the analysis of the excess classification error in our theorem to that of the $L^\infty$ norm.
In Appendix~\ref{sec:required_number_of_optimized_random_features_for_function_approximation}, we analyze the required number of optimized RFs for our function approximation.
In Appendix~\ref{sec:translation_of_l_2_distance_into_l_infty_distance}, we show how to obtain the bound in terms of the $L^\infty$ norm from the $L^2$ norm used in the analysis of optimized RFs.
Along with these analysis on optimized RFs, we also analyze stochastic gradient descent (SGD) used for our algorithm in Appendix~\ref{sec:analysis_of_stochastic_gradient_descent}.
Using these results, we present a proof of Theorem~\ref{sthm:generalization} in Appendix~\ref{sec:proof_of_theorem_1}.
In Appendix~\ref{sec:feasibility}, we also explain a feasible implementation of the input model of the quantum algorithm of~\citet{NEURIPS2020_9ddb9dd5} in our setting.

We repeat our algorithm in Algorithm~\ref{salg:classification} for readability and show a formal statement of the main theorem in the main text as Theorem~\ref{sthm:generalization} in the following.
In our analysis, we will write our estimate of the function to be learned as
\begin{equation}
  \label{seq:estimate}
  \hat{f}_{v,\alpha}(x)=\sum_{m=0}^{M-1}(\alpha_{2m}\cos(-2\pi v_m\cdot x)+\alpha_{2m+1}\sin(-2\pi v_m\cdot x))\in\mathcal{F}_M,
\end{equation}
where $\mathcal{F}_M$ is the RKHS associated with the approximated kernel $k_M$ in term of RFs used in place of $k$, i.e.,
\begin{equation}
  \label{seq:k_M}
  k_M(x,x^\prime)=\sum_{m=0}^{M-1}(\nicefrac{1}{(Mq_\lambda^\ast(v_m))})\overline{\varphi(v_m,x)}\varphi(v_m,x^\prime)\quad\text{with $\varphi(v,x)=\mathrm{e}^{-2\pi\mathrm{i}v\cdot x}$}.
\end{equation}
We let $\|\hat{f}_{v,\alpha}\|_{\mathcal{F}_M}$ denote the RKHS norm of $\hat{f}_{v,\alpha}$ in $\mathcal{F}_M$.

\begin{algorithm}[h]
  \caption{\label{salg:classification}Classification with optimized RFs.}
  \begin{algorithmic}[1]
    \REQUIRE{Parameter $\lambda$ in~\eqref{seq:bound_lambda_thm}, number $M$ of optimized RFs in~\eqref{seq:bound_lambda_epsilon_thm}, even number $N\in 2\mathbb{N}$ of labeled examples in~\eqref{seq:bound_T_epsilon_thm}, unlabeled input examples stored in the data structure described in the main text, step sizes $\eta_t$ in~\eqref{seq:eta}, parameter $q_{\min}$ in~\eqref{seq:q_min}, parameter region $\mathcal{W}\subset\mathbb{R}^{2M}$.}
    \ENSURE{A classifier $\hat{f}_{v,\alpha}$ in the form of~\eqref{seq:estimate} with $v_0,\ldots,v_{M-1}\in\mathcal{V}=\mathbb{R}^D$ and $\alpha_0,\ldots,\alpha_{2M-1}\in\mathbb{R}$ that achieves~\eqref{seq:expected_excess_testing_error}, i.e., the excess classification error bounded by $\epsilon$ in expectation, with exponential error convergence shown in Theorem~\ref{sthm:generalization} and within a remarkably short runtime $T_\mathrm{all}$ shown in the main text.}
    \FOR{$m\in\left\{0,\ldots,M-1\right\}$}
    \STATE{Sample an optimized RF $v_m\in\mathcal{V}$ according to the optimized distribution $q_\lambda^\ast(v)d\tau(v)$ in~\eqref{seq:q_complex}.}
    \COMMENT{See \textit{Sampling optimized RFs} in Sec.~3.1 of the main text for detail.}
    \ENDFOR%
    \STATE{Initialize $\alpha^{(0)}\gets{(0,0,\ldots,0)}\in\mathcal{W}$.}
    \COMMENT{See \textit{SGD} in Sec.~3.1 of the main text for detail.}
    \FOR{$t\in\left\{0,\ldots,N-1\right\}$}
    \STATE{Calculate the prefactor $C(\alpha^{(t)})$ of the unbiased estimate $\hat{g}^{\left(t\right)}$ of the gradient of the regularized testing loss $\cL_\lambda$ in~\eqref{seq:regularized_loss_alpha} at $\alpha^{(t)}$ according to~\eqref{seq:prefactor}.}
    \STATE{Calculate $\hat{g}^{\left(t\right)}$ according to~\eqref{seq:unbiased_estimate} using $C(\alpha^{(t)})$.}
    \STATE{Set $\alpha^{\left(t+1\right)}\gets\Pi_\mathcal{W}(\alpha^{\left(t\right)}-\eta^{\left(t\right)} \hat{g}^{\left(t\right)})$.}
    \COMMENT{Using projection $\Pi_\mathcal{W}$ onto $\mathcal{W}$.}
    \ENDFOR%
    \STATE{Set ${(\alpha_0,\ldots,\alpha_{2M-1})}\gets(\nicefrac{2}{N})\sum_{t=\nicefrac{N}{2}+1}^{N}\alpha^{(t)}$.}
    \COMMENT{Suffix averaging.}
    \STATE{\textbf{Return } $v_0,\ldots,v_{M-1}\in\mathcal{V}$ and $\alpha_0,\ldots,\alpha_{2M-1}\in\mathbb{R}$.}
    \COMMENT{i.e., $\hat{f}_{v,\alpha}$ with these parameters.}
  \end{algorithmic}
\end{algorithm}

\begin{theorem}[\label{sthm:generalization}Generalization property of optimized RFs in Algorithm~\ref{salg:classification}]
  Fix arbitrarily small $p\in(0,1)$.
  There exist the parameter $\lambda$, the number $M$ of optimized RFs and the number $N$ of labeled examples satisfying
  \begin{align}
    \label{seq:bound_lambda_thm}
    \lambda&=O\left(\frac{\delta^{2}}{\|f^\ast\|_\mathcal{F}^2}{\left(\frac{\delta}{\|f^\ast\|_{\mathcal{F}}\sqrt{q_{\min}}}\right)}^{-\frac{2p}{1+p}}\right),\\
    \label{seq:bound_lambda_epsilon_thm}
    M&=O\left(d\left(\lambda\right)\log\left(\frac{d\left(\lambda\right)}{\epsilon}\right)\right),\\
    \label{seq:bound_T_epsilon_thm}
    N&=O\left(\log\left(\frac{1}{\epsilon}\right)\frac{\|f^\ast\|_\mathcal{F}^4}{\delta^{4} q_{\min}^2}{\left(\frac{\|f^\ast\|_{\mathcal{F}}}{\lambda\delta\sqrt{q_{\min}}}\right)}^{\frac{4p}{1-p}}\right),
  \end{align}
  such that Algorithm~\ref{sthm:generalization} can return a classifier $\hat{f}_{v,\alpha}$ with the excess classification error bounded by
  \begin{equation}
    \label{seq:expected_excess_testing_error}
    \mathbb{E}\left[\cR(\hat{f}_{v,\alpha})-\cR^\ast\right]\leqq\epsilon,
  \end{equation}
  where $C>0$ is a constant independent of $p$, $d(\lambda)$ is defined as
  \begin{equation}
    \label{seq:degree_of_freedom}
    d(\lambda)\coloneqq\tr\Sigma{(\Sigma+\lambda\mathbbm{1})}^{-1},
  \end{equation}
  $\|f^\ast\|_\mathcal{F}$ is the RKHS norm of the optimal classifier in the RKHS $\mathcal{F}$ of the Gaussian kernel $k$
  \begin{equation}
    \label{seq:f_optimal}
    f^\ast(x)\coloneqq\mathbb{E}[Y|X=x]=\rho(1|x)-\rho(-1|x)\in\mathcal{F}.
  \end{equation}
\end{theorem}

\section{Reduction of analysis of excess classification error to $L^\infty$ norm}%
\label{sec:reduction_to_l_infty_norm}

As discussed in Sec.~3.2 of the main text,
we reduce the analysis of the excess classification error~\eqref{seq:expected_excess_testing_error} to an evaluation of the $L^\infty$ norm.
Whereas some existing analysis of SGD under the low-noise condition
\begin{equation}
  \label{seq:low_noise}
  \left|\mathbb{E}[Y|X]\right|>\delta\quad\text{almost surely for some $0<\delta\leqq 1$},
\end{equation}
by~\citet{pmlr-v75-pillaud-vivien18a,pmlr-v89-nitanda19a} have used a reduction of the analysis of~\eqref{seq:expected_excess_testing_error} to that of the RKHS norm $\|\cdot\|_\mathcal{F}$, this reduction is not straightforwardly applicable to our estimate $\hat{f}\in\mathcal{F}_M$ since $\mathcal{F}_M$ is not necessarily included in $\mathcal{F}$.
Instead, the $L^\infty$ norm is used for the analysis of SGD under the low-noise condition by~\citet{Yashima2019}, but~\citet{Yashima2019} shows a statement on the number $M$ of RFs in high probability, and conditioned on this, a statement on the number of iterations in SGD (i.e., the number $N$ of labeled examples in our setting) in (conditional) expectation.
In contrast, we here want to show statements in expectation as in~\eqref{seq:expected_excess_testing_error}, rather than the statement in high probability and conditional expectation.
We here show the following proposition.
Note that the following proposition is written in terms of $\hat{f}\in\mathcal{F}_M$ and $f^\ast\in\mathcal{F}$ for clarity but can be generalized to $\hat{f},f^\ast\in L^\infty(d\rho_\cX)$ using the same proof.

\begin{proposition}[\label{sprp:l_infty_bound}Reduction of analysis of excess classification error to $L^\infty$ norm]
  Let $\hat{f}\in\mathcal{F}_M$ be a random function obtained by a randomized algorithm.
  Given the function $f^\ast\in\mathcal{F}$ defined as~\eqref{seq:f_optimal},
  if it holds with high probability greater than $1-\epsilon$ that
  \begin{equation}
    \label{seq:l_infty_condition}
    \|\hat{f}-f^\ast\|_{L^\infty(d\rho_\cX)}<\delta,
  \end{equation}
  then we have
  \begin{equation}
    \mathbb{E}\left[\cR(\hat{f})-\cR^\ast\right]\leqq\epsilon.
  \end{equation}
\end{proposition}

\begin{proof}
  The proof follows from a similar argument to the proof of Lemma~1 in the work by~\citet{pmlr-v75-pillaud-vivien18a}.
  Due to the low-noise condition~\eqref{seq:low_noise} and the definition~\eqref{seq:f_optimal} of $f^\ast$, if $\hat{f}$ satisfies~\eqref{seq:l_infty_condition}, then we have for almost all $x\in\cX_\rho$
  \begin{equation}
    \sign(\hat{f}(x))=\sign(f^\ast(x)).
  \end{equation}
  Then, for $\hat{f}$ satisfying~\eqref{seq:l_infty_condition}, it holds that
  \begin{equation}
    \cR(\hat{f})-\cR^\ast=0,
  \end{equation}
  which holds with probability greater than $1-\epsilon$.
  Since we have $\cR(\hat{f})-\cR^\ast\leqq 1$ for any $\hat{f}$ not necessarily satisfying~\eqref{seq:l_infty_condition}, it holds that
  \begin{equation}
    \mathbb{E}\left[\cR(\hat{f})-\cR^\ast\right]\leqq(1-\epsilon)\times 0+\epsilon\times 1\leqq\epsilon.
  \end{equation}
\end{proof}

\section{Required number of optimized random features for function approximation}%
\label{sec:required_number_of_optimized_random_features_for_function_approximation}

Due to Proposition~\ref{sprp:l_infty_bound}, to achieve the goal~\eqref{seq:expected_excess_testing_error} of our classification task, it suffices to obtain $\hat{f}_{v,\alpha}(x)$ that has the form of~\eqref{seq:estimate} and satisfies~\eqref{seq:l_infty_condition} with the high probability;
to bound the required number $M$ of optimized RFs for satisfying this, as discussed in Sec.~3.2 of the main text,
we here analyze the distance in the $L^2$ norm between $\hat{f}_{v,\alpha}\in\mathcal{F}_M$ and $f^\ast\in\mathcal{F}$, which will be translated to the $L^\infty$ norm used in~\eqref{seq:l_infty_condition} later in Appendix~\ref{sec:translation_of_l_2_distance_into_l_infty_distance}.
In the following, we summarize the results by~\citet{B1} on the required number $M$ of optimized RFs for function approximation, and then proceed to present our results on bounding $M$ and also $\|\alpha\|_2^2$, where $\alpha={(\alpha_0,\ldots,\alpha_{2M-1})}\in\mathbb{R}^{2M}$ is the coefficients in~\eqref{seq:estimate}.

The results by~\citet{B1} can be formulated in terms of a real-valued feature map rather than the complex-valued feature map used in the main text, i.e.,
\begin{equation}
  \label{seq:varphi_defn}
  \varphi(v,x)=\mathrm{e}^{-2\pi\mathrm{i}v\cdot x},
\end{equation}
where $v\in\mathcal{V}=\mathbb{R}^D$, $x\in\mathcal{X}=\mathbb{R}^D$, and $D$ is the dimension of input data.
To represent the real-valued kernel $k$,
instead of complex-valued $\varphi(v,x)$,
it is also possible to use a real-valued feature map $\varphi(v,b,\cdot)$ with a parameter space $(v,b)\in\mathcal{V}\times\mathcal{B}$ given by~\citep{B1,R2}
\begin{equation}
  \label{seq:varphi_defn_real}
  \varphi(v,b,x)\coloneqq\sqrt{2}\cos(-2\pi v\cdot x+2\pi b),
\end{equation}
where the space for
\begin{equation}
  b\in\mathcal{B}\coloneqq[0,1]
\end{equation}
is equipped with the uniform distribution.
We write
\begin{equation}
  \label{seq:tau_v_b}
  d\tau(v,b)\coloneqq d\tau(v)\,db.
\end{equation}
Indeed, the kernel $k$ can be represented by
\begin{align}
  k(x,x^\prime)&=\int_\mathcal{V}d\tau(v)\overline{\varphi(v,x)}\varphi(v,x^\prime)\\
               &=\int_\mathcal{V}d\tau(v)(\cos(-2\pi v\cdot x)\cos(-2\pi v\cdot x^\prime)+\sin(-2\pi v\cdot x)\sin(-2\pi v\cdot x^\prime))\\
               &=\int_\mathcal{V}d\tau(v)\int_\mathcal{B}db \sqrt{2}\cos(-2\pi v\cdot x+2\pi b)\times\sqrt{2}\cos(-2\pi v\cdot x^\prime+2\pi b)\\
               &=\int_\mathcal{V}d\tau(v)\int_\mathcal{B}db \varphi(v,b,x)\varphi(v,b,x^\prime).
\end{align}
Using the real-valued feature maps, we can represent any function $f\in\mathcal{F}$ in the RKHS associated with $k$ as~\citep{B1}
\begin{equation}
  f(x)=\Braket{\alpha(\cdot,\cdot)|\varphi(\cdot,\cdot,x)}_{L^2(d\tau(v,b))},
\end{equation}
where $\alpha(v,b)$ is a function satisfying $\|\alpha\|_{L^2(d\tau)}<\infty$ and can be regarded as coefficients of the Fourier-basis function $\varphi(\cdot,\cdot,x)$.
In this case of using real-valued feature maps,
conventional algorithms using random features by~\citet{R2,R3} start with sampling $M$ $(D+1)$-dimensional parameters
\begin{equation}
  (v_0,b_0),\ldots,(v_{M-1},b_{M-1})\in\mathcal{V}\times\mathcal{B}
\end{equation}
from the distribution $d\tau(v,b)=d\tau(v)db$ corresponding to the kernel $k$,
so as to determine $M$ features $\varphi(v_m,b_m\cdot)$.
Using these features, the kernel can be approximated as~\citep{R2}
\begin{align}
  k(x,x^\prime)&=\int_\mathcal{V}d\tau(v)\overline{\varphi(v,x)}\varphi(v,x^\prime)\nonumber\\
               &=\int_\mathcal{V}d\tau(v)\int_\mathcal{B}db{\varphi(v,b,x)}\varphi(v,b,x^\prime)\nonumber\\
               &\approx\frac{1}{M}\sum_{m=0}^{M-1}\varphi(v_m,b_m,x)\varphi(v_m,b_m,x^\prime).
\end{align}
Moreover, $f\in\mathcal{F}$ can be approximated as~\citep{R3}
\begin{align}
  f(x)&=\int_\mathcal{V}d\tau(v,b)\alpha(v,b)\varphi(v,b,x)\\
            &\approx\sum_{m=0}^{M-1}\alpha_m\varphi(v_m,b_m,x),
\end{align}
where $\alpha_0,\ldots,\alpha_{M-1}\in\mathbb{R}$ are some coefficients.
As summarized in the main text,
to achieve the learning to accuracy $O({\epsilon})$, we need to sample a sufficiently large number $M$ of features.
Once we fix $M$ features, we calculate coefficients $\alpha_m$ by linear (or ridge) regression using the given examples~\citep{R3,carratino2018learning,NIPS2017_6914}.
Then, to minimize $M$,~\citet{B1} provides an optimized probability density function $q_{\lambda}^\ast(v,b)$ for $d\tau(v,b)$ given by
\begin{align}
  \label{seq:q}
  q_{\lambda}^\ast\left(v,b\right)\propto{\braket{\varphi\left(v,b,\cdot\right)|{\left(\Sigma+{\epsilon}\mathbbm{1}\right)}^{-1}\varphi\left(v,b,\cdot\right)}_{L^2(d\rho_\cX)}},
\end{align}
where $\lambda$ is a parameter for regularization, and $\Sigma:L^2(d\rho_\cX)\to L^2(d\rho_\cX)$ is the integral operator~\citep{C2}
\begin{equation}
  \left(\Sigma f\right)\left(x^\prime\right)\coloneqq\int_\mathcal{X}d\rho\left(x\right)\,k\left(x^\prime,x\right)f\left(x\right).
\end{equation}
The normalization yields
\begin{equation}
  \int_\mathcal{V}d\tau(v)\int_\mathcal{B} db q_{\lambda}^\ast\left(v,b\right)=1.
\end{equation}
As shown in the following lemma,~\citet{B1} shows that it suffices to sample $M$ features from $q_{\lambda}^\ast(v)d\tau(v)$ with $M$ bounded by
\begin{equation}
  \label{seq:M_bound}
  M=O\left(d\left({\lambda}\right)\log\left(\frac{d\left({\lambda}\right)}{\epsilon}\right)\right),
\end{equation}
so as to achieve the learning to accuracy $O(\lambda)$ with high probability greater than $1-\epsilon$,
where $d\left({\lambda}\right)$ is given by~\eqref{seq:degree_of_freedom}.

\begin{lemma}[\label{slem:bach}\citet{B1}: Real-valued optimized random features]
  For any $\lambda>0$ and any $\epsilon\in(0,1)$,
  let $(v_0,b_0),\ldots,(v_{M-1},b_{M-1})\in\mathcal{V}\times\mathcal{B}$ be sampled IID from the density $q_\lambda^\ast(v,b)$ with respect to $d\tau(v,b)$.
  If $M$ satisfies
  \begin{equation}
    \label{seq:M_achievablity}
    M\geqq 5d(\lambda)\ln\left(\frac{16d(\lambda)}{\epsilon}\right),
  \end{equation}
  then, with high probability greater than $1-\epsilon$, it holds for any $f\in\mathcal{F}$ that
  \begin{equation}
    \label{seq:error_bound}
    \min_{\hat{f}}\|\hat{f}-f\|_{L^2(d\rho_\cX)}^2\leqq 4\lambda\|f\|_\mathcal{F}^2,
  \end{equation}
  where the minimum is taken over all the functions
  \begin{equation}
    \label{seq:hat_f}
    \hat{f}=\sum_{m=0}^{M-1}\frac{\beta_m}{\sqrt{Mq_\lambda^\ast(v_m,b_m)}}\varphi(v_m,b_m,\cdot)
  \end{equation}
  with coefficients $\beta={(\beta_0,\ldots,\beta_{M-1})}\in\mathbb{R}^M$ satisfying
  \begin{equation}
    \label{seq:constraint}
    \|\beta\|_2\leqq2\|f\|_\mathcal{F}.
  \end{equation}
\end{lemma}

As discussed in Sec.~3.2 of the main text, the problem of Lemma~\ref{slem:bach} arises from the fact that the values of $q_\lambda^\ast(v_0,b_0),\ldots,q_\lambda^\ast(v_{M-1},b_{M-1})$ in~\eqref{seq:hat_f} are unknown throughout Algorithm~\ref{salg:classification},
and we here show how to address this problem.
Our approach is to derive a new constraint on the coefficients in place of~\eqref{seq:constraint}, so that the dependency on $v_m$ and $b_m$ can be eliminated.
To eliminate the dependency on $v_m$,
we can use a hyperparameter $q_{\min}$ in place of the values of $q_\lambda^\ast$, i.e.,
\begin{equation}
  \label{seq:q_min}
  q_{\min}\coloneqq\min\{q_\lambda^\ast(v_m):m=0,\ldots,M-1\}.
\end{equation}
As for eliminating the dependency on $b_m$,
in contrast to Lemma~\ref{slem:bach} using the sampling from $q_\lambda^\ast(v,b)d\tau(v,b)$ defined in terms of a real-valued feature $\varphi(v,b,x)$ given by~\eqref{seq:varphi_defn_real},
we use in Algorithm~\ref{salg:classification} the optimized distribution
\begin{equation}
  \label{seq:q_complex}
  q_\lambda^\ast(v)d\tau(v)\propto\braket{\varphi(v,\cdot)|{(\Sigma+\lambda\mathbbm{1})}^{-1}\varphi(v,\cdot)}_{L_2(d\rho_\mathcal{X})}d\tau(v),
\end{equation}
defined in terms of a complex-valued feature $\varphi(v,x)$ given by~\eqref{seq:varphi_defn}.
Note that this is feasible in our setting since the quantum algorithm of~\citet{NEURIPS2020_9ddb9dd5} indeed performs sampling from $q_\lambda^\ast(v)d\tau(v)$ defined in terms of the complex-valued feature.
As expected, we here argue that the difference between the real-valued and complex-valued features in defining $q_\lambda^\ast$ does not affect the applicability of our algorithm to the learning of real-valued functions.
To see this, observe that $q_\lambda^\ast(v)d\tau(v)$ is a marginal distribution of $q_\lambda^\ast(v,b)d\tau(v,b)$, i.e.,
\begin{align}
  &q_\lambda^\ast(v)=\int_\mathcal{B}db\, q_\lambda^\ast(v,b)\\
  &\propto{\braket{\cos\left(-2\pi v\cdot(\cdot)\right)|{\left(\Sigma+{\lambda}\mathbbm{1}\right)}^{-1}\cos\left(-2\pi v\cdot(\cdot)\right)}_{L^2(d\rho_\cX)}}\\
  &\quad+{\braket{\sin\left(-2\pi v\cdot(\cdot)\right)|{\left(\Sigma+{\lambda}\mathbbm{1}\right)}^{-1}\sin\left(-2\pi v\cdot(\cdot)\right)}_{L^2(d\rho_\cX)}}.
\end{align}
That is, sampling $v\in\mathcal{V}$ from $q_\lambda^\ast(v)d\tau(v)$ followed by sampling $b\in\mathcal{B}$ from a conditional distribution
\begin{equation}
  \label{seq:q_lambda_conditioned}
  q_\lambda^\ast(b|v)db\coloneqq\frac{q_\lambda^\ast(v,b)}{q_\lambda^\ast(v)}db
\end{equation}
is exactly equivalent to sampling $(v,b)\in\mathcal{V}\times\mathcal{B}$ from $q_\lambda^\ast(v,b)d\tau(v,b)$.
However, whatever $b\in\mathcal{B}$ is sampled, we can expand each term in the representation~\eqref{seq:hat_f} of $\hat{f}$ as
\begin{align}
  \frac{\beta_m}{\sqrt{Mq_\lambda^\ast(v_m)}}\varphi(v_m,b_m,x)&=\frac{\beta_m}{\sqrt{Mq_\lambda^\ast(v_m)}}\sqrt{2}\cos(-2\pi v_m\cdot x+2\pi b_m)\\
  \label{seq:cos_sin}
                                                                   &=\alpha_{2m}\cos(-2\pi v_m\cdot x)+\alpha_{2m+1}\sin(-2\pi v_m\cdot x),
\end{align}
by appropriately choosing the coefficients $\alpha_{2m},\alpha_{2m+1}\in\mathbb{R}$ for each $m\in\{0,\ldots,M-1\}$.
As a result,~\eqref{seq:cos_sin} becomes independent of $b_m$.
Based on these observations, we show the following proposition in place of Lemma~\ref{slem:bach}.

\begin{proposition}[\label{sprp:complex_optimized_random_feature}Complex-valued optimized random features]
  For any $\lambda>0$ and any $\epsilon\in(0,1)$, let $v_0,\ldots,v_{M-1}\in\mathcal{V}$ be sampled IID from the density $q_\lambda^\ast(v)$ with respect to $d\tau(v)$.
  If $M$ satisfies
  \begin{equation}
    \label{seq:M_achievablity_complex}
    M\geqq 5d(\lambda)\ln\left(\frac{16d(\lambda)}{\epsilon}\right),
  \end{equation}
  then, with high probability greater than $1-\epsilon$, it holds for any $f\in\mathcal{F}$ that
  \begin{equation}
    \label{seq:error_bound_complex}
    \min_{\hat{f}_{v,\alpha}}\|\hat{f}_{v,\alpha}-f\|_{L^2(d\rho_\cX)}^2\leqq 4\lambda\|f\|_\mathcal{F}^2,
  \end{equation}
  where the minimum is taken over all the functions $\hat{f}_{v,\alpha}(x)\in\mathcal{F}_M$ in the form of~\eqref{seq:estimate}
  with coefficients $\alpha={(\alpha_0,\ldots,\alpha_{2M-1})}\in\mathbb{R}^{2M}$ satisfying
  \begin{equation}
    \label{seq:constraint_complex}
    \|\alpha\|_2\leqq\frac{2\sqrt{2}\|f\|_\mathcal{F}}{\sqrt{Mq_{\min}}},
  \end{equation}
  and $q_{\min}$ is a constant given by~\eqref{seq:q_min}.
\end{proposition}

\begin{proof}
  Lemma~\ref{slem:bach} shows that if we sample $v_m\in\mathcal{V}$ according to $q_\lambda^\ast(v)d\tau(v)$ and sample $b_m\in\mathcal{B}$ according to $q_\lambda^\ast(b|v_m)db$ in~\eqref{seq:q_lambda_conditioned} for each $m\in\{0,\ldots,M-1\}$,
  then, with high probability greater than $1-\epsilon$, we have the bound~\eqref{seq:error_bound} with
  \begin{align}
    \hat{f}(x)&=\sum_{m=0}^{M-1}\frac{\beta_m}{\sqrt{Mq_\lambda^\ast(v_m,b_m)}}\varphi(v_m,b_m,x)\\
           &=\sum_{m=0}^{M-1}\frac{\beta_m}{\sqrt{Mq_\lambda^\ast(v_m,b_m)}}\sqrt{2}\cos(-2\pi v_m\cdot x+2\pi b_m)\\
           &=\sum_{m=0}^{M-1}\left(\frac{\beta_m}{\sqrt{Mq_\lambda^\ast(v_m,b_m)}}\sqrt{2}\cos(2\pi b_m)\cos(-2\pi v_m\cdot x)\right.\\
           &\quad\left.-\frac{\beta_m}{\sqrt{Mq_\lambda^\ast(v_m,b_m)}}\sqrt{2}\sin(2\pi b_m)\sin(-2\pi v_m\cdot x)\right).
  \end{align}
  where $\beta$ satisfies the constraint~\eqref{seq:constraint}.
  Choosing coefficients $\alpha_{2m},\alpha_{2m+1}$ for each $m\in\{0,\ldots,M-1\}$ as
  \begin{align}
    \alpha_{2m}&\coloneqq\frac{\beta_m}{\sqrt{Mq_\lambda^\ast(v_m)}}\sqrt{2}\cos(2\pi b_m),\\
    \alpha_{2m+1}&\coloneqq-\frac{\beta_m}{\sqrt{Mq_\lambda^\ast(v_m)}}\sqrt{2}\sin(2\pi b_m),
  \end{align}
  we have $\hat{f}_{v,\alpha}$ in the form of~\eqref{seq:estimate} that satisfies the bound~\eqref{seq:error_bound_complex}.
  The constraint~\eqref{seq:constraint_complex} of $\alpha$ is obtained from the constraint~\eqref{seq:constraint} of $\beta$ by
  \begin{align}
    \|\alpha\|_2&=\sqrt{\sum_{m=0}^{2M-1}\alpha_m^2}\\
                &=\sqrt{\sum_{m=0}^{M-1}\left(\frac{\beta_m^2}{Mq_\lambda^\ast(v_m,b_m)}2\cos^2(2\pi b_m)+\frac{\beta_m^2}{Mq_\lambda^\ast(v_m,b_m)}2\sin^2(2\pi b_m)\right)}\\
                &=\sqrt{\sum_{m=0}^{M-1}\frac{2\beta_m^2}{Mq_\lambda^\ast(v_m,b_m)}}\\
                &\leqq\sqrt{\sum_{m=0}^{M-1}\frac{2\beta_m^2}{Mq_{\min}}}\\
                &=\frac{\sqrt{2}\|\beta\|_2}{\sqrt{Mq_{\min}}}\\
                &\leqq\frac{2\sqrt{2}\|f\|_\mathcal{F}}{\sqrt{Mq_{\min}}}.
  \end{align}
\end{proof}

\section{Translation from $L^2$ norm into $L^\infty$ norm}%
\label{sec:translation_of_l_2_distance_into_l_infty_distance}

As discussed in Sec.~3.2 of the main text, we here show how to translate the $L^2$ norm used so far into a bound in the $L^\infty$ norm used in the left-hand side of~\eqref{seq:l_infty_condition}.
To bound the $L^\infty$ norm on the left-hand side of~\eqref{seq:l_infty_condition} by the $L^2$ norm,
we use the fact that the kernel is Gaussian, $\supp(\rho_\cX)\subset\mathbb{R}^D$ is bounded, and $\rho_\cX$ has a density that is uniformly bounded away from $0$ and $\infty$ on $\supp(\rho_\cX)$.
In this setting,~\citet{Yashima2019} shows an upper bound of the $L^\infty$ norm in terms of the RKHS norm in $\mathcal{F}$ and the $L^2$ norm; in particular, for any $p\in(0,1)$, there exists a constant $C_p>0$ such that it holds for any $f\in\mathcal{F}$ that
\begin{align}
  \label{seq:kernel_condition}
  \|f\|_{L^\infty(d\rho_\cX)}\leqq C_p\|f\|_\mathcal{F}^p\|f\|_{L^2(d\rho_\cX)}^{1-p}.
\end{align}
Although~\eqref{seq:kernel_condition} is in terms of the RKHS norm in $\mathcal{F}$,~\citet{Yashima2019} also shows a similar bound that holds with high probability greater than $1-\epsilon$ for functions represented in terms of RFs sampled from the data-independent distribution $d\tau(v)$; however, this bound includes a polynomially large factor in $\nicefrac{1}{\epsilon}$.
We here improve and generalize this bound; significantly, our bound depends only poly-logarithmically on $\nicefrac{1}{\epsilon}$, unlike that shown by~\citet{Yashima2019}.
Consequently, our result here achieves an exponential improvement in $\nicefrac{1}{\epsilon}$, and this improvement will be crucial in the proof of Theorem~\ref{sthm:generalization} not to cancel out the exponential convergence of classification error in SGD under the low-noise condition shown by~\citet{pmlr-v75-pillaud-vivien18a,pmlr-v89-nitanda19a,Yashima2019}.

To show our result here, recall that the function $\hat{f}_{v,\alpha}$ in~\eqref{seq:estimate} used for our function approximation is in the RKHS $\mathcal{F}_M$ associated with a kernel $k_M$ in~\eqref{seq:k_M},
and $\mathcal{F}_M$ is not necessarily included in $\mathcal{F}$ associated with the kernel $k$.
The RHKSs $\mathcal{F}$ and $\mathcal{F}_M$ are characterized by~\citep{B1}
\begin{align}
  \mathcal{F}&=\left\{\int_\mathcal{V}d\tau(v)\int_\mathcal{B}db \beta(v,b)\varphi(v,b,\cdot):\beta\in L^2(d\tau)\right\},\\
  \mathcal{F}_M&=\left\{\sum_{m=0}^{M-1}\frac{\beta_{m}}{\sqrt{q_\lambda^\ast(v_m,b_m)}}\varphi(v_m,b_m,\cdot):\|\beta\|_2<\infty\right\}.
\end{align}
For $f\in\mathcal{F}$ and $\hat{f}\in\mathcal{F}_M$, the norms in these RKHS are given by~\citep{B1}
\begin{align}
  \|f\|_\mathcal{F}&=\inf\left\{{\|\beta\|}_{L^2(d\tau(v,b))}:f=\int_\mathcal{V}d\tau(v)\int_\mathcal{B}db \beta(v,b)\varphi(v,b,\cdot)\right\},\\
  \label{seq:rkhs_norm_F_M}
  \|\hat{f}\|_{\mathcal{F}_M}&=\inf\left\{{\|\beta\|}_2:\hat{f}=\sum_{m=0}^{M-1}\frac{\beta_{m}}{\sqrt{Mq_\lambda^\ast(v_m)}}\varphi(v_m,b_m,\cdot)\right\}.
\end{align}
To deal with functions in $\mathcal{F}$ and $\mathcal{F}_M$ at the same time,
let $\mathcal{F}_M^+$ be the RKHS associated with a kernel
\begin{equation}
 \label{seq:kernel_f_M_plus}
  k+k_M,
\end{equation}
so that, for any $f\in\mathcal{F}$ and $\hat{f}\in\mathcal{F}_M$, we have $f,\hat{f}\in\mathcal{F}_M^+$~\citep{Steinwart2008}.
The RKHS norm of $f\in\mathcal{F}_M^+$ is given by
\begin{equation}
  \label{seq:rkhs_norm_f_M_plus}
  \|f\|_{\mathcal{F}_M^+}=\inf\left\{{\|f_1\|}_{\mathcal{F}}+{\|f_2\|}_{\mathcal{F}_M}:f=f_1+f_2,f_1\in\mathcal{F},f_2\in\mathcal{F}_M\right\}.
\end{equation}
Then, in place of~\eqref{seq:kernel_condition}, we show a bound in terms of the RKHS norm in $\mathcal{F}_M^+$ and the $L^2$ norm that is applicable to functions represented by optimized RFs sampled from the optimized distribution $q_\lambda^\ast(v)d\tau(v)$, as shown in the following proposition.

\begin{proposition}[\label{sprp:Gaussian}Bound of $L^\infty$ norm in terms of RKHS norm and $L^2$ norm]
  Fix an arbitrarily small $p\in(0,1)$.
  For any $\epsilon>0$,
  with high probability greater than $1-\epsilon$ in sampling $M$ random features from the optimized distribution $q_\lambda^\ast(v)d\tau(v)$,
  it holds for any $f\in\mathcal{F}_M^+$ that
  \begin{align}
    \label{seq:l_infty_condition_in_l_2}
    \|f\|_{L^\infty(d\rho)}=O\left({\left(1+\sqrt{\frac{1}{M}\log\left(\frac{1}{\epsilon}\right)}\right)}^{\frac{p}{2}}\|f\|_{\mathcal{F}_M^+}^{p}\|f\|_{L^2(d\rho)}^{1-p}\right).
  \end{align}
  where the constant factor depends on $p$.
\end{proposition}

\begin{proof}
  Our proof follows from a similar argument to the proof of Theorem~2 of the work by~\citet{Yashima2019}, based on embedding of $\mathcal{F}_M^+$ in the Sobolev space.
  Let $W^m(\cX_\rho)$ denote a Sobolev space of order $m$, i.e.,
  \begin{equation}
    W^m(\cX_\rho)\coloneqq\left\{f\in L^2(\cX_\rho):\text{$\partial^{(\alpha)}f\in L^2(\cX_\rho)$ exists for all $\alpha\in\mathbb{N}^D$ with $|\alpha|\leqq m$.}\right\},
  \end{equation}
  where $\partial^{(\alpha)}$ is the $\alpha$th weak derivative for a multiindex $\alpha=(\alpha^{(1)},\ldots,\alpha^{(D)})\in\mathbb{N}^D$ with  $|\alpha|=\sum_{d=1}^D\alpha^{(d)}$.
  Choose an arbitrary integer $m\geqq\nicefrac{D}{2}$, and fix $\alpha=(\alpha^{(1)},\ldots,\alpha^{(D)})\in\mathbb{N}^D$ with $|\alpha|=\sum_{d=1}^D\alpha^{(d)}=m$.

  \citet{Yashima2019} shows that there exists a constant $C_1>0$ such that
  \begin{equation}
    \label{seq:1}
    \|f\|_{L^\infty(\cX_\rho)}\leqq C_1\|f\|_{W^m(\cX)}^{\nicefrac{d}{2m}}\|f\|_{L^2(d\rho_\mathcal{X})}^{1-\nicefrac{d}{2m}}.
  \end{equation}
  Since $\cX_\rho$ is bounded, we have
  \begin{equation}
    \label{seq:2}
    \|f\|_{L^\infty(d\rho)}=\|f\|_{L^\infty(\cX_\rho)}.
  \end{equation}
  Under our assumption that $\rho$ is uniformly bounded away from $\infty$ on $\cX_\rho$, there exists a constant $C_2>0$ such that
  \begin{equation}
    \label{seq:3}
    \|f\|_{L^2(d\rho)}=C_2\|f\|_{L^2(\cX_\rho)}.
  \end{equation}
  Thus, from~\eqref{seq:1},~\eqref{seq:2}, and~\eqref{seq:3}, we obtain
  \begin{equation}
    \label{seq:4}
    \|f\|_{L^\infty(d\rho)}\leqq C_3\|f\|_{W^m(\cX)}^{\nicefrac{d}{2m}}\|f\|_{L^2(\cX_\rho)}^{1-\nicefrac{d}{2m}},
  \end{equation}
  where $C_3>0$ is a constant.

  Then, to bound $\|f\|_{W^m(\cX)}$ in~\eqref{seq:4} in terms of the RKHS norm of $\mathcal{F}_M^+$ in~\eqref{seq:rkhs_norm_f_M_plus},
  we define
  \begin{equation}
    \tau^+(v)\coloneqq\sum_{m=0}^{M-1}\frac{1}{Mq_\lambda^{\ast}(v_m)}\delta(v-v_m)+\tau(v),
  \end{equation}
  where $\delta$ is the delta function, and $q_\lambda^{\ast}(v_m)$ appears in the denominator due to the importance sampling of $v_m$ according to $q_\lambda^{\ast}(v_m)d\tau(v_m)$~\citep{B1}.
  Similarly to~\eqref{seq:tau_v_b}, we write
  \begin{equation}
    d\tau^+(v,b)=d\tau^+(v)\,db,
  \end{equation}
  where $b\in\mathcal{B}$ is used for the real-valued feature map $\varphi(v,b,x)$ given by~\eqref{seq:varphi_defn_real}.
  The kernel $k+k_M$ in~\eqref{seq:kernel_f_M_plus} is written as
  \begin{equation}
    (k+k_M)(x,x^\prime)=\int_{\mathcal{V}}d\tau^+(v)\overline{\varphi(v,x)}\varphi(v,x^\prime),
  \end{equation}
  and for any $f\in\mathcal{F}_M^+$, there exists $g\in L^2(d\tau^+(v,b))$ such that~\citep{B1}
  \begin{align}
    f(x)&=\int_{\mathcal{V}\times\mathcal{B}}d\tau^+(v,b) g(v,b)\varphi(v,b,x),\\
    \|f\|_{\mathcal{F}_M^+}&=\|g\|_{L^2(d\tau^+(v,b))}.
  \end{align}
  For $\alpha={(\alpha^{(1)},\ldots,\alpha^{(D)})}\in\mathbb{N}^D$ and $v={(v^{(1)},\ldots,v^{(D)})}\in\mathcal{V}$, we write
  \begin{align}
    v^\alpha&=\prod_{d=1}^D{\left(v^{(d)}\right)}^{\alpha^{(d)}},\\
    \partial^\alpha&=\partial_1^{\alpha^{(1)}}\cdots\partial_D^{\alpha^{(D)}}.
  \end{align}
  Then, we have
  \begin{align}
    \|\partial^\alpha f\|_{L^2(\cX_\rho)}^2&=\int_{\cX_\rho}dx\,{\left(\partial^\alpha_x\int_{\mathcal{V}\times\mathcal{B}}d\tau^+(v,b) g(v,b)\varphi(v,b,x)\right)}^2\\
                                           &\leqq\int_{\cX_\rho}dx\,{\left(\int_{\mathcal{V}\times\mathcal{B}}d\tau^+(v,b) |g(v,b)|\partial^\alpha_x\varphi(v,b,x)\right)}^2\\
                                           &\leqq\|g\|_{L^2(d\tau^+(v,b))}^2\int_{\cX_\rho}dx\,\int_{\mathcal{V}\times\mathcal{B}}d\tau^+(v,b) {\left|\partial^\alpha_x\varphi(v,b,x)\right|}^2\\
                                           \label{seq:C_4}
                                           &\leqq\|g\|_{L^2(d\tau^+(v,b))}^2\int_{\cX_\rho}dx\,\int_{\mathcal{B}}db\,\int_{\mathcal{V}}d\tau^+(v) C_4v^{2\alpha}\\
                                           \label{seq:C_4_2}
                                           &\leqq C_4 \mathrm{vol}(\cX_\rho)\|f\|_{\mathcal{F}_M^+}^2\left(\mathbb{E}_{v\sim\tau}[v^{2\alpha}]+\sum_{m=0}^{M-1}\frac{1}{Mq_\lambda^\ast(v_m)}v_m^{2\alpha}\right),
  \end{align}
  where~\eqref{seq:C_4} with a constant factor $C_4>0$ follows from the fact that $\varphi$ is a random Fourier feature, i.e., $\sin$ and $\cos$.
  Since $\tau$ is Gaussian for the Gaussian kernel,
  the term $\mathbb{E}_{v\sim\tau}[v^{2\alpha}]$ in~\eqref{seq:C_4_2} is finite, i.e.,
  \begin{equation}
    \mathbb{E}_{v\sim\tau}[v^{2\alpha}]<\infty.
  \end{equation}
  Moreover, to bound the other term in~\eqref{seq:C_4_2}, i.e.,
  \begin{equation}
    \sum_{m=0}^{M-1}\frac{1}{Mq_\lambda^\ast(v_m)}v_m^{2\alpha},
  \end{equation}
  for $m\in\{0,\ldots,M-1\}$, define an IID random variable
  \begin{equation}
    X_m\coloneqq\frac{1}{q_\lambda^\ast(v_m)}v_m^{2\alpha}.
  \end{equation}
  Recall that the importance sampling of $v_m$ is performed according to $q_\lambda^\ast(v_m)d\tau(v_m)$, where the distribution $\tau$ is Gaussian; in addition, we have $2|\alpha|=2\sum_{d=1}^D\alpha^{(d)}=2m\geqq D\geqq 1$.
  Hence, $X_m$ is a sub-Gaussian random variable.
  Therefore, Hoeffding's inequality for the sub-Gaussian random variable yields~\citep{vershynin_2018}
  \begin{equation}
    \label{seq:concentration}
    \Pr\left\{\frac{1}{M}\sum_{m=0}^{M-1}X_m\geqq t\right\}\leqq 2\exp\left(-\frac{C_5t^2}{\frac{1}{M^2}\sum_{m=0}^{M-1}\|X_m\|_{\psi_2}^2}\right),
  \end{equation}
  where $C_5>0$ is a constant, and $\|X\|_{\psi_2}$ is given by
  \begin{equation}
    \|X\|_{\psi_2}\coloneqq\inf\left\{t>0:\mathbb{E}\left[\exp\left(\frac{X^2}{t^2}\right)\leqq 2\right]\right\}.
  \end{equation}
  Considering the variance of the Gaussian distribution $\tau$ to be constant, we have
  \begin{equation}
    \frac{1}{M^2}\sum_{m=0}^{M-1}\|X_m\|_{\psi_2}^2=O\left(\frac{1}{M}\right).
  \end{equation}
  Thus, with probability greater than $1-\epsilon$, it follows from~\eqref{seq:concentration} that
  \begin{equation}
    \frac{1}{M}\sum_{m=0}^{M-1}X_m\leqq O\left(\sqrt{\frac{1}{M}\log\left(\frac{1}{\epsilon}\right)}\right).
  \end{equation}
  Therefore, with probability greater than $1-\epsilon$, we obtain a bound of $\|f\|_{W^m(\cX)}$ in~\eqref{seq:4} given by
  \begin{equation}
    \|f\|_{W^m(\cX_\rho)}^2=O\left(\left(1+\sqrt{\frac{1}{M}\log\left(\frac{1}{\epsilon}\right)}\right)\|f\|_{\mathcal{F}_M^+}^2 \right),
  \end{equation}
  where $C_4$, $\mathrm{vol}(\cX_\rho)$, and $\mathbb{E}_{v\sim\tau}\left[v^{2\alpha}\right]$ in~\eqref{seq:4} are included in a constant factor.

  Consequently, by taking $p=\nicefrac{d}{2m}$, we have
  \begin{equation}
    \|f\|_{L^\infty(d\rho)}=O\left({\left(1+\sqrt{\frac{1}{M}\log\left(\frac{1}{\epsilon}\right)}\right)}^{\frac{p}{2}}\|f\|_{\mathcal{F}_M^+}^{p}\|f\|_{L^2(d\rho)}^{1-p}\right).
  \end{equation}
\end{proof}

\section{Approximation bound for stochastic gradient descent}%
\label{sec:analysis_of_stochastic_gradient_descent}

In this section, we analyze the SGD in Algorithm~\ref{salg:classification}.
Proposition~\ref{sprp:Gaussian} shows that we can bound the $L^\infty$ norm in~\eqref{seq:l_infty_condition} of Proposition~\ref{sprp:l_infty_bound} by bounding the $L^2$ norm on the left-hand side of~\eqref{seq:l_infty_condition_in_l_2}, and hence toward our goal of bounding the $L^\infty$ norm, it suffices to analyze the minimization of the loss in terms of the $L^2$ norm by the SGD\@.
Based on the approximation of functions with optimized random features in Proposition~\ref{sprp:complex_optimized_random_feature},
let $\hat{f}^\ast$ be a minimizer of the loss $\cL$, i.e.,
\begin{align}
\label{seq:hat_f_ast}
\hat{f}^\ast\coloneqq\argmin_{\hat{f}}&\Big\{\cL(\alpha):\hat{f}(x)=\sum_{m=0}^{M-1}\left(\alpha_{2m}\cos(-2\pi v_m\cdot x)+\alpha_{2m+1}\sin(-2\pi v_m\cdot x)\right),\nonumber\\
&\|\alpha\|_2\leqq\frac{2\sqrt{2}\|f^\ast\|_\mathcal{F}}{\sqrt{Mq_{\min}}}\Big\}\in\mathcal{F}_M,
\end{align}
where $\cL$ is defined as
\begin{equation}
  \label{seq:loss_alpha}
  \cL(\alpha)=\cL(\hat{f}_{v,\alpha})\coloneqq\mathbb{E}\Big[\Big|\sum_{m=0}^{M-1}\big(\alpha_{2m}\cos(-2\pi v_m\cdot X)+\alpha_{2m+1}\sin(-2\pi v_m\cdot X)\big)-Y\Big|^2\Big].
\end{equation}
Also let $\hat{f}^\ast_\lambda$ be a minimizer of the regularized loss $\cL_\lambda$, i.e.,
\begin{align}
\label{seq:hat_f_ast_lambda}
\hat{f}^\ast_\lambda\coloneqq\argmin_{\hat{f}}&\Big\{\cL_\lambda(\alpha):\hat{f}(x)=\sum_{m=0}^{M-1}\left(\alpha_{2m}\cos(-2\pi v_m\cdot x)+\alpha_{2m+1}\sin(-2\pi v_m\cdot x)\right),\nonumber\\
                                              &\|\alpha\|_2\leqq\frac{2\sqrt{2}\|f^\ast\|_\mathcal{F}}{\sqrt{Mq_{\min}}}\Big\}\in\mathcal{F}_M,
\end{align}
where $\cL_\lambda$ is defined as
\begin{equation}
  \label{seq:regularized_loss_alpha}
  \cL_\lambda(\alpha)=\cL_\lambda(\hat{f}_{v,\alpha})\coloneqq\cL(\alpha)+\lambda Mq_{\min}\|\alpha\|_2^2.
\end{equation}
We bound $\|\hat{f}^\ast_\lambda-f^\ast\|_{L^2(d\rho_\cX)}$ in the $L^2$ norm as follows.

\begin{proposition}[\label{sthm:M}Bound on $\|\hat{f}^\ast_\lambda-f^\ast\|_{L^2(d\rho_\cX)}$]
  For any $\lambda>0$,
  if $M$ satisfies the condition~\eqref{seq:M_achievablity_complex} in Proposition~\ref{sprp:complex_optimized_random_feature},
  then it holds with probability greater than $1-\epsilon$ that
  \begin{equation}
    \|\hat{f}^\ast_\lambda-f^\ast\|_{L^2(d\rho_\cX)}^2+\lambda Mq_{\min}\|\alpha\|_2^2\leqq 12\lambda\|f^\ast\|_\mathcal{F}^2,
  \end{equation}
  where $\alpha$ is the coefficients of $\hat{f}^\ast_\lambda$ in~\eqref{seq:hat_f_ast_lambda}.
\end{proposition}

\begin{proof}
  Due to $f^\ast\in\mathcal{F}$, Proposition~\ref{sprp:complex_optimized_random_feature} shows that with probability greater than $1-\epsilon$, there exists
  \begin{equation}
    \hat{f}(x)=\sum_{m=0}^{M-1}\left(\alpha_{2m}\cos\left(-2\pi v_m\cdot x\right)+\alpha_{2m+1}\sin\left(-2\pi v_m\cdot x\right)\right)\in\mathcal{F}_M
  \end{equation}
  such that
  \begin{align}
    \label{seq:L_2_bound}
    \|\hat{f}-f^\ast\|_{L^2(d\rho_\cX)}^2&\leqq 4\lambda\|f^\ast\|_\mathcal{F}^2,\\
    \label{seq:condition_hat_f_norm}
    \|\alpha\|_2&\leqq\frac{2\sqrt{2}\|f^\ast\|_\mathcal{F}}{\sqrt{Mq_{\min}}}.
  \end{align}
  By definition~\eqref{seq:hat_f_ast} of $\hat{f}^\ast$,
  the function $\hat{f}^\ast$ minimizes
  \begin{align}
    \label{seq:minimizer_loss}
    &\mathbb{E}\left[\left|\hat{f}(X)-Y\right|^2\right]\nonumber\\
    &=\int_{\mathcal{X}\times\mathcal{Y}}d\rho(x,y){\left(\hat{f}(x)-y\right)}^2\nonumber\\
    &=\int_{\mathcal{X}\times\mathcal{Y}}d\rho(x,y){\left({(\hat{f}(x))}^2-2\hat{f}(x)y+y^2\right)}\nonumber\\
    &=\int_{\mathcal{X}}d\rho_\mathcal{X}(x)\left({(\hat{f}(x))}^2-2\hat{f}(x)(\rho(1|x)-\rho(-1|x))\right)+{\left(\int_{\mathcal{X}\times\mathcal{Y}}d\rho(x,y)y^2\right)}\nonumber\\
    &=\int_{\mathcal{X}}d\rho_\mathcal{X}(x)\left({(\hat{f}(x))}^2-2\hat{f}(x)f^\ast(x)\right)+\int_{\mathcal{X}\times\mathcal{Y}}d\rho(x,y)y^2,
  \end{align}
  where the second last line follows from the definition~\eqref{seq:f_optimal} of $f^\ast$.
  Since the last term $\int_{\mathcal{X}\times\mathcal{Y}}d\rho(x,y)y^2$ is a constant, $\hat{f}^\ast$ is also a minimizer of
  \begin{align}
    \label{seq:minimizer_loss_ast}
    &\left\|\hat{f}-f^\ast\right\|_{L^2(d\rho_\cX)}^2\\
    &=\mathbb{E}\left[\left|\hat{f}(X)-f^\ast(X)\right|^2\right]\\
    &=\int_{\mathcal{X}}d\rho_\mathcal{X}(x)\left({(\hat{f}(x))}^2-2\hat{f}(x)f^\ast(x)\right)+{\left(\int_{\mathcal{X}}d\rho_\cX(x){(f^\ast(x))}^2\right)}.
  \end{align}
  Therefore, due to~\eqref{seq:L_2_bound}, we have
  \begin{equation}
    \|\hat{f}^\ast-f^\ast\|_{L^2(d\rho_\cX)}^2\leqq \|\hat{f}-f^\ast\|_{L^2(d\rho_\cX)}^2\leqq 4\lambda\|f^\ast\|_\mathcal{F}^2.
  \end{equation}

  Consequently, using the same observation as~\eqref{seq:minimizer_loss} and~\eqref{seq:minimizer_loss_ast}, the minimizer $\hat{f}^\ast_\lambda$ in~\eqref{seq:hat_f_ast_lambda} of the regularized testing loss satisfies
  \begin{align}
    \label{seq:hat_f_ast_lambda_f_ast_bound}
    &\|\hat{f}^\ast_\lambda-f^\ast\|_{L^2(d\rho_\cX)}^2+\lambda Mq_{\min}\|\alpha(\hat{f}^\ast_\lambda)\|_2^2\\
    &\leqq\|\hat{f}^\ast-f^\ast\|_{L^2(d\rho_\cX)}^2+\lambda Mq_{\min}\|\alpha(\hat{f}^\ast)\|_2^2\\
    &\leqq 4\lambda\|f^\ast\|_\mathcal{F}^2+8\lambda\|f^\ast\|_\mathcal{F}^2\\
    &=12\lambda\|f^\ast\|_\mathcal{F}^2,
  \end{align}
  where $\alpha(\hat{f}^\ast_\lambda)$ and $\alpha(\hat{f}^\ast)$ are coefficients $\alpha$ for $\hat{f}^\ast_\lambda$ in~\eqref{seq:hat_f_ast_lambda} and $\hat{f}^\ast$ in~\eqref{seq:hat_f_ast}, respectively.
  Thus, we have the conclusion.
\end{proof}

Moreover, we show that we can use the bound on $\|\hat{f}^\ast_\lambda-f^\ast\|_{L^2(d\rho_\cX)}$ in Proposition~\ref{sthm:M} to bound $\|\hat{f}_{v,\alpha}-f^\ast\|_{L^2(d\rho_\cX)}$,
where $\hat{f}_{v,\alpha}$ is a function in the form of~\eqref{seq:estimate} described by $v_0,\ldots,v_{M-1}$ and $\alpha_0,\ldots,\alpha_{2M-1}$ obtained from Algorithm~\ref{salg:classification}.
The analysis is based on the bound shown by~\citet{H3} on an estimate of the solution $\hat{f}^\ast_\lambda\in\mathcal{F}_M$ in the minimization of~\eqref{seq:hat_f_ast_lambda} by the SGD in Algorithm~\ref{salg:classification} after $N$ iterations.

\begin{proposition}[\label{sthm:T}Bound on $\|\hat{f}_{v,\alpha}-f^\ast\|_{L^2(d\rho_\cX)}$]
  Let $\hat{f}_{v,\alpha}\in\mathcal{F}_M$ be a function in the form of~\eqref{seq:estimate} described by $v_0,\ldots,v_{M-1}$ and $\alpha_0,\ldots,\alpha_{2M-1}$ obtained from Algorithm~\ref{salg:classification}.
  Then,
  it holds with high probability greater than $1-\epsilon$ that
  \begin{equation}
    \label{seq:condition_hat_f_ast}
    \|\hat{f}_{v,\alpha}-f^\ast\|_{L^2(d\rho_\cX)}=O\left(\sqrt{\frac{1}{N}\times\frac{\|f^\ast\|_\mathcal{F}^2}{\lambda q_{\min}^2}\log\left(\frac{1}{\epsilon}\right)+12\lambda\|f^\ast\|_\mathcal{F}^2}\right).
  \end{equation}
\end{proposition}

\begin{proof}
  To analyze the left-hand side of~\eqref{seq:condition_hat_f_ast}, in place of the regularized testing loss $\cL_\lambda(\alpha)$ in~\eqref{seq:regularized_loss_alpha}, we use another function
  \begin{equation}
    \cL^\prime_\lambda(\alpha)\coloneqq\left\|\sum_{m=0}^{M-1}\left(\alpha_{2m}\cos(-2\pi v_m\cdot (\cdot))+\alpha_{2m+1}\sin(-2\pi v_m\cdot (\cdot))\right)-f^\ast\right\|_{L^2(d\rho_\cX)}^2+\lambda Mq_{\min}\|\alpha\|_2^2.
  \end{equation}
  As observed in~\eqref{seq:minimizer_loss} and~\eqref{seq:minimizer_loss_ast},
  $\cL^\prime_\lambda(\alpha)$ is different from $\cL_\lambda(\alpha)$ just by a constant term.
  Hence, $\hat{f}^\ast_\lambda\in\mathcal{F}_M$ defined in~\eqref{seq:hat_f_ast_lambda} as a minimizer of the regularized testing loss $\cL_\lambda$ is also a minimizer of $\cL^\prime_\lambda(\alpha)$, that is,
  \begin{align}
    \hat{f}^\ast_\lambda=\argmin_{\hat{f}}&\Big\{\cL^\prime_\lambda(\alpha):\hat{f}=\sum_{m=0}^{M-1}\left(\alpha_{2m}\cos(-2\pi v_m\cdot (\cdot))+\alpha_{2m+1}\sin(-2\pi v_m\cdot (\cdot))\right),\\
                        &\|\alpha\|_2\leqq\frac{2\sqrt{2}\|f\|_\mathcal{F}}{\sqrt{Mq_{\min}}}\Big\}\in\mathcal{F}_M.
  \end{align}
  Moreover, the gradients of $\cL_\lambda(\alpha)$ and $\cL^\prime_\lambda(\alpha)$ coincide; that is, for any $\alpha\in\mathbb{R}^{2M}$, we have
  \begin{equation}
    \nabla\cL^\prime_\lambda(\alpha)=\nabla\cL_\lambda(\alpha),
  \end{equation}
  and hence,
  with
  \begin{equation}
    \label{seq:prefactor}
    C\Big(\alpha^{(t)}\Big)\coloneqq 2\left(\sum_{m=0}^{M-1}\left(\alpha_{2m}^{(t)}\cos(-2\pi v_m\cdot x_t)+\alpha_{2m+1}^{(t)}\sin(-2\pi v_m\cdot x_t)\right)-y_t\right),
  \end{equation}
  the unbiased estimate $\hat{g}^{(t)}$ of $\nabla\cL_\lambda(\alpha)$
  \begin{equation}
    \label{seq:unbiased_estimate}
    \hat{g}^{(t)}=C\left(\alpha^{(t)}\right)\left(\begin{matrix}
        \cos(-2\pi v_0\cdot x_t)\\
        \sin(-2\pi v_0\cdot x_t)\\
        \vdots\\
        \cos(-2\pi v_{M-1}\cdot x_t)\\
        \sin(-2\pi v_{M-1}\cdot x_t)
        \end{matrix}\right)+2\lambda Mq_{\min}\left(\begin{matrix}
        \alpha^{(t)}_0\\
        \alpha^{(t)}_1\\
        \vdots\\
        \alpha^{(t)}_{2M-2}\\
        \alpha^{(t)}_{2M-1}
    \end{matrix}\right)
    \in\mathbb{R}^{2M}
  \end{equation}
  also serves as an unbiased estimate of $\nabla\cL^\prime_\lambda(\alpha^{(t)})$ for each $t\in\{0,\ldots,N-1\}$.
  Thus, the SGD in Algorithm~\ref{salg:classification} is equivalent to minimization of $\cL^\prime_\lambda(\alpha)$.

  We here bound
  \begin{equation}
    \cL^\prime_\lambda(\alpha)-\cL^\prime_\lambda(\alpha^\ast),
  \end{equation}
  where $\alpha$ in the following is the coefficient after the $N$ iterations of the SGD in Algorithm~\ref{salg:classification}, and we let $\alpha^\ast\in\mathcal{W}$ denote the coefficients of $\hat{f}^\ast_\lambda$.
  To bound this, we use an upper bound of the number of iterations in SGD with suffix averaging given by~\citet{H3} in terms of high probability, which shows that if we have the following:
  \begin{itemize}
    \item $\cL^\prime_\lambda$ is $\mu$-strongly convex;
    \item for any $\alpha\in\mathcal{W}$,
      \begin{align}
        \left\|\nabla \cL^\prime_\lambda(\alpha)\right\|_2\leqq L;
      \end{align}
    \item the unbiased estimate $\hat{g}$ for any point $\alpha\in\mathcal{W}$ almost surely satisfies
      \begin{align}
        \left\|\hat{g}\right\|_2\leqq L;
      \end{align}
  \end{itemize}
  then, after $N$ iterations, with high probability greater than $1-\epsilon$, the final $N$th iterate of the SGD in Algorithm~\ref{salg:classification} returns $\alpha$ satisfying
  \begin{equation}
    \label{seq:bound_iteration}
    \cL^\prime_\lambda(\alpha)-\cL^\prime_\lambda(\alpha^\ast)=O\left(\frac{L^2}{\mu}\frac{\log\left(\frac{1}{\epsilon}\right)}{N}\right).
  \end{equation}
  In our case, due to the regularization term $\lambda Mq_{\min}\|\alpha\|_2^2$ of $\cL^\prime_\lambda$, $\mu$ can be given by
  \begin{equation}
    \label{seq:mu}
    \mu=\lambda Mq_{\min}.
  \end{equation}
  Note that for minimization of a $\mu$-strongly convex function,~\citet{H3} assumes to use step size scaling as
  \begin{equation}
    \label{seq:eta}
    \eta^{(t)}=O\left(\nicefrac{1}{\mu t}\right),
  \end{equation}
  which is indeed used in Algorithm~\ref{salg:classification}.
  As for $L$, the constraint~\eqref{seq:constraint_complex} of $\alpha$ yields
  \begin{align}
    \|\nabla\cL^\prime(\alpha^{(t)})\|_2=\|\nabla\cL(\alpha^{(t)})\|_2&= O\left(\sqrt{\sum_{m=0}^{2M-1}{\left(\sum_{m^\prime=0}^{2M-1}|\alpha_{m^\prime}|+\lambda Mq_{\min}\alpha_m\right)}^2}\right)\\
                                                                      &= O\left(\sqrt{M\|\alpha\|_1^2+\lambda Mq_{\min}\|\alpha\|_1^2+{(\lambda Mq_{\min})}^2\|\alpha\|_2^2}\right)\\
                                                                      &=O\left(\sqrt{M^2\|\alpha\|_2^2}\right)=O\left(\sqrt{M}\frac{\|f\|_\mathcal{F}}{\sqrt{q_{\min}}}\right),\\
    \|\hat{g}^{(t)}\|_2&=O\left(\sqrt{M^2\|\alpha\|_2^2}\right)=O\left(\sqrt{M}\frac{\|f\|_\mathcal{F}}{\sqrt{q_{\min}}}\right),
  \end{align}
  where we use $\|\alpha\|_1^2\leqq 2M\|\alpha\|_2^2$ and the fact that we consider the case of $\lambda q_{\min}=O(1)$.
  Thus, we have
  \begin{equation}
    \label{seq:L}
    L=O\left(\sqrt{M}\frac{\|f\|_\mathcal{F}}{\sqrt{q_{\min}}}\right).
  \end{equation}
  Consequently, from~\eqref{seq:mu} and~\eqref{seq:L}, we obtain
  \begin{equation}
    \cL^\prime(\alpha)-\cL^\prime(\alpha^\ast)=O\left(\frac{\log\left(\frac{1}{\epsilon}\right)}{N}\times\frac{\|f\|_\mathcal{F}^2}{\lambda q_{\min}^2}\right),
  \end{equation}
  that is,
  \begin{align}
    \label{seq:hat_f_T_f_ast_bound}
    &\|\hat{f}_{v,\alpha}-f^\ast\|_{L^2(d\rho_\cX)}^2+\lambda Mq_{\min}\|\alpha(\hat{f}_{v,\alpha})\|_2^2\\
    &\leqq\frac{1}{N}\times\frac{\|f\|_\mathcal{F}^2}{\lambda q_{\min}^2}\log\left(\frac{1}{\epsilon}\right)+\|\hat{f}^\ast_\lambda-f^\ast\|_{L^2(d\rho_\cX)}^2+\lambda Mq_{\min}\|\alpha(\hat{f}^\ast_\lambda)\|_2^2\\
    &\leqq\frac{1}{N}\times\frac{\|f\|_\mathcal{F}^2}{\lambda q_{\min}^2}\log\left(\frac{1}{\epsilon}\right)+12\lambda\|f^\ast\|_\mathcal{F}^2,
  \end{align}
  where $\alpha(\hat{f}_{v,\alpha})$ and $\alpha(\hat{f}^\ast_\lambda)$ in the same way as~\eqref{seq:hat_f_ast_lambda_f_ast_bound} are coefficients $\alpha$ of $\hat{f}_{v,\alpha}$ and $\hat{f}^\ast_\lambda$, respectively, and the last inequality follows from Proposition~\ref{sthm:M}.
  Therefore, we have
  \begin{equation}
    \|\hat{f}_{v,\alpha}-f^\ast\|_{L^2(d\rho_\cX)}=O\left(\sqrt{\frac{1}{N}\times\frac{\|f\|_\mathcal{F}^2}{\lambda q_{\min}^2}\log\left(\frac{1}{\epsilon}\right)+12\lambda\|f^\ast\|_\mathcal{F}^2}\right).
  \end{equation}
\end{proof}

Consequently, by combining Propositions~\ref{sprp:complex_optimized_random_feature},~\ref{sprp:Gaussian} and~\ref{sthm:T}, we obtain the following bound in terms of the $L^\infty$ norm.

\begin{proposition}[\label{sthm:T_infty}Bound on $\|\hat{f}_{v,\alpha}-f^\ast\|_{L^\infty(d\rho_\cX)}$]
  Fix an arbitrarily small $p\in(0,1)$.
  For any $\epsilon>0$,
  it holds with high probability greater than $1-\epsilon$ that
  \begin{align}
    \|\hat{f}_{v,\alpha}-f^\ast\|_{L^\infty(d\rho_\cX)}=O\left({\left(\|\hat{f}_{v,\alpha}\|_{\mathcal{F}_M}+\|f^\ast\|_{\mathcal{F}}\right)}^{p}{\left(\frac{1}{N}\times\frac{\|f\|_\mathcal{F}^2}{\lambda q_{\min}^2}\log\left(\frac{1}{\epsilon}\right)+12\lambda\|f^\ast\|_{\mathcal{F}}^2\right)}^{\frac{1-p}{2}}\right).
  \end{align}
\end{proposition}

\begin{proof}
  Due to Proposition~\ref{sprp:Gaussian}, it holds with probability greater than
  $1-\epsilon$ that
  \begin{align}
    &\|\hat{f}_{v,\alpha}-f^\ast\|_{L^\infty(d\rho_\cX)}\\
    &=O\left({\left(1+\sqrt{\frac{1}{M}\log\left(\frac{1}{\epsilon}\right)}\right)}^{\frac{p}{2}}\|\hat{f}_{v,\alpha}-f^\ast\|_{\mathcal{F}_M^+}^{p}\|\hat{f}_{v,\alpha}-f^\ast\|_{L^2(d\rho)}^{1-p}\right)\\
    \label{seq:1_}
    &=O\left({\left(1+\sqrt{\frac{1}{M}\log\left(\frac{1}{\epsilon}\right)}\right)}^{\frac{p}{2}}{\left(\|\hat{f}_{v,\alpha}\|_{\mathcal{F}_M}+\|f^\ast\|_{\mathcal{F}}\right)}^{p}\|\hat{f}_{v,\alpha}-f^\ast\|_{L^2(d\rho)}^{1-p}\right),
  \end{align}
  where~\eqref{seq:1_} follows from~\eqref{seq:rkhs_norm_f_M_plus}.
  As a result,
  with probability greater than $1-\epsilon$,
  we obtain from Proposition~\ref{sthm:T}
  \begin{align}
    &\|\hat{f}_{v,\alpha}-f^\ast\|_{L^\infty(d\rho_\cX)}=\\
    &O\left({\left(1+\sqrt{\frac{1}{M}\log\left(\frac{1}{\epsilon}\right)}\right)}^{\frac{p}{2}}{\left(\|\hat{f}_{v,\alpha}\|_{\mathcal{F}_M}+\|f^\ast\|_{\mathcal{F}}\right)}^{p}{\left(\frac{1}{N}\times\frac{\|f\|_\mathcal{F}^2}{\lambda q_{\min}^2}\log\left(\frac{1}{\epsilon}\right)+12\lambda\|f^\ast\|_{\mathcal{F}}^2\right)}^{\frac{1-p}{2}}\right).
  \end{align}
  This bound holds for any $M$ satisfying the condition~\eqref{seq:M_achievablity_complex} in Proposition~\ref{sprp:complex_optimized_random_feature}, i.e.,
  \begin{equation}
    M=\Omega\left(d(\lambda)\log\left(\frac{d(\lambda)}{\epsilon}\right)\right).
  \end{equation}
  In this case, we have
  \begin{align}
    \label{seq:M_loose_bound}
    &{\left(1+\sqrt{\frac{1}{M}\log\left(\frac{1}{\epsilon}\right)}\right)}^{\frac{p}{2}}\\
    &=O\left({\left(1+\sqrt{\frac{1}{d(\lambda)\log\left(\frac{d(\lambda)}{\epsilon}\right)}\log\left(\frac{1}{\epsilon}\right)}\right)}^{\frac{p}{2}}\right)\\
    &=O\left({\left(1+\sqrt{\frac{1}{d(\lambda)}}\right)}^{\frac{p}{2}}\right)\\
    &=O\left(1\right),
  \end{align}
  which yields
  \begin{align}
    &\|\hat{f}_{v,\alpha}-f^\ast\|_{L^\infty(d\rho_\cX)}\\
    &=O\left({\left(\|\hat{f}_{v,\alpha}\|_{\mathcal{F}_M}+\|f^\ast\|_{\mathcal{F}}\right)}^{p}{\left(\frac{1}{N}\times\frac{\|f^\ast\|_\mathcal{F}^2}{\lambda q_{\min}^2}\log\left(\frac{1}{\epsilon}\right)+12\lambda\|f^\ast\|_{\mathcal{F}}^2\right)}^{\frac{1-p}{2}}\right).
  \end{align}
\end{proof}

\section{Proof of Theorem~1}%
\label{sec:proof_of_theorem_1}

Using Propositions~\ref{sprp:l_infty_bound} and~\ref{sthm:T_infty},
we show Theorem~\ref{sthm:generalization} on the generalization property of optimized RF in Algorithm~\ref{salg:classification}.

\begin{proof}[Proof of Theorem~\ref{sthm:generalization}]
  We analyze $\lambda$ and $N$ since the bound of $M$ immediately follows from Proposition~\ref{sprp:complex_optimized_random_feature}.
  To bound the left-hand side of~\eqref{seq:expected_excess_testing_error} using Proposition~\ref{sprp:l_infty_bound},
  we bound
  \begin{equation}
    \|\hat{f}_{v,\alpha}-f^\ast\|_{L^\infty(d\rho_\cX)},
  \end{equation}
  where $\hat{f}_{v,\alpha}$ is the function obtained by Algorithm~\ref{salg:classification} after $N$ iterations of the SGD as in Proposition~\ref{sthm:T_infty}.
  Then, as we have shown in Proposition~\ref{sthm:T_infty}, it holds with probability greater than $1-\epsilon$ that
  \begin{align}
    \label{seq:inf_norm_bound}
    &\|\hat{f}_{v,\alpha}-f^\ast\|_{L^\infty(d\rho_\cX)}=O\left({\left(\|\hat{f}_{v,\alpha}\|_{\mathcal{F}_M}+\|f^\ast\|_{\mathcal{F}}\right)}^{p}{\left(\frac{1}{N}\times\frac{\|f^\ast\|_\mathcal{F}^2}{\lambda q_{\min}^2}\log\left(\frac{1}{\epsilon}\right)+12\lambda\|f^\ast\|_{\mathcal{F}}^2\right)}^{\frac{1-p}{2}}\right).
  \end{align}
  To achieve the assumption~\eqref{seq:l_infty_condition} of Proposition~\ref{sprp:l_infty_bound}, i.e.,
  \begin{equation}
    \label{seq:bound_infinity_norm}
    \|\hat{f}_{v,\alpha}-f^\ast\|_{L^\infty(d\rho_\cX)}<\delta,
  \end{equation}
  we want to set $\lambda$ and $N$ in such a way that we have
  \begin{align}
    \label{seq:bound_lambda}
    \lambda&=O\left(\frac{\delta^{2}}{\|f^\ast\|_\mathcal{F}^2}{\left(\frac{\|\hat{f}_{v,\alpha}\|_{\mathcal{F}_M}+\|f^\ast\|_{\mathcal{F}}}{\delta}\right)}^{-\frac{2p}{1-p}}\right),\\
    \label{seq:bound_T}
    N&=O\left(\frac{\|f^\ast\|_\mathcal{F}^2\log\left(\frac{1}{\epsilon}\right)}{\delta^{2}\lambda q_{\min}^2}{\left(\frac{\|\hat{f}_{v,\alpha}\|_{\mathcal{F}_M}+\|f^\ast\|_{\mathcal{F}}}{\delta}\right)}^{\frac{2p}{1-p}}\right)\\
     &=O\left(\frac{\|f^\ast\|_\mathcal{F}^4\log\left(\frac{1}{\epsilon}\right)}{\delta^{4} q_{\min}^2}{\left(\frac{\|\hat{f}_{v,\alpha}\|_{\mathcal{F}_M}+\|f^\ast\|_{\mathcal{F}}}{\delta}\right)}^{\frac{4p}{1-p}}\right).
  \end{align}
  If these inequalities are satisfied, then Proposition~\ref{sprp:l_infty_bound} yields
  \begin{equation}
    \mathbb{E}\left[\cR(\hat{f})-\cR^\ast\right]\leqq\epsilon,
  \end{equation}
  as desired.
  To complete the proof, we need to eliminate $\|\hat{f}_{v,\alpha}\|_{\mathcal{F}_M}$ appearing in these inequalities.

  We here prove that $\|\hat{f}_{v,\alpha}\|_{\mathcal{F}_M}$ in~\eqref{seq:inf_norm_bound} does not blow up to infinity and indeed can be upper bounded explicitly.
  By definition~\eqref{seq:rkhs_norm_F_M}, we have
  \begin{equation}
    \|\hat{f}_{v,\alpha}\|_{\mathcal{F}_M}\leqq\sqrt{\sum_{m=0}^{M-1}Mq_\lambda^\ast(v_m)\alpha_m^2}.
  \end{equation}
  Similarly to $q_{\min}$ in~\eqref{seq:q_min}, we here write
  \begin{equation}
    q_{\max}\coloneqq\max\{q_\lambda^\ast(v_m):M=0,\ldots,M-1\},
  \end{equation}
  so the right-hand side of the last inequality can be bounded by
  \begin{equation}
    \sqrt{\sum_{m=0}^{M-1}Mq_\lambda^\ast(v_m)\alpha_m^2}\leqq\sqrt{\sum_{m=0}^{M-1}Mq_{\max}\alpha_m^2}.
  \end{equation}
  Due to~\eqref{seq:constraint_complex}, the right-hand side can be bounded by
  \begin{equation}
    \sqrt{\sum_{m=0}^{M-1}Mq_{\max}\alpha_m^2}\leqq 2\sqrt{2}\|f^\ast\|_{\mathcal{F}}\sqrt{\frac{q_{\max}}{q_{\min}}}.
  \end{equation}
  An upper bound of $q_{\max}$ is
  \begin{equation}
    \label{seq:q_max_bound}
    q_{\max}\leqq\max\{q_\lambda^\ast(v):v\in\mathcal{V}\},
  \end{equation}
  where the optimized distribution $q_\lambda^\ast(v)$ is defined as~\citep{B1}
  \begin{equation}
    q_\lambda^\ast(v)=\frac{\Braket{\varphi(v,\cdot)|{(\Sigma+\lambda\mathbbm{1})}^{-1}\varphi(v,\cdot)}_{L^2(d\rho_\mathcal{X})}}{d(\lambda)}.
  \end{equation}
  Since the largest eigenvalue of ${(\Sigma+\lambda\mathbbm{1})}^{-1}$ is upper bounded by $\nicefrac{1}{\lambda}$, we can evaluate the upper bound of $q_{\max}$ in~\eqref{seq:q_max_bound} as
  \begin{equation}
\label{seq:worst_case_bound}
    q_{\max}\leqq\frac{\nicefrac{1}{\lambda}}{d(\lambda)}=O\left(\frac{1}{\lambda}\right),
  \end{equation}
  where the right-hand side is obtained by considering the worst-case bound, i.e., the case where $d(\lambda)$ is as small as constant.
  Note that even in our advantageous cases discussed in the main text, we indeed have a better bound
  \begin{equation}
    d(\lambda)=O(\polylog(\nicefrac{1}{\lambda})),
  \end{equation}
  and then $q_{\max}$ is bounded by
  \begin{equation}
    q_{\max}=O\left(\frac{1}{\lambda\polylog(\nicefrac{1}{\lambda})}\right),
  \end{equation}
  while the use of the worst-case bound~\eqref{seq:worst_case_bound} suffices for the proof here.
  Thus, $\|\hat{f}_{v,\alpha}\|_{\mathcal{F}_M}$ is bounded by
  \begin{equation}
    \label{seq:bound_f}
    \|\hat{f}_{v,\alpha}\|_{\mathcal{F}_M}\leqq 2\sqrt{2}\|f^\ast\|_{\mathcal{F}}\sqrt{\frac{q_{\max}}{q_{\min}}}=O\left(\frac{\|f^\ast\|_{\mathcal{F}}}{\lambda\sqrt{q_{\min}}}\right).
  \end{equation}

  Consequently, due to~\eqref{seq:bound_f}, we can fulfill~\eqref{seq:bound_lambda} with $\lambda$ satisfying
  \begin{equation}
    \lambda^{\frac{1+p}{1-p}}=O\left(\frac{\delta^{2}}{\|f^\ast\|_\mathcal{F}^2}{\left(\frac{\|f^\ast\|_{\mathcal{F}}}{\delta\sqrt{q_{\min}}}\right)}^{-\frac{2p}{1-p}}\right),
  \end{equation}
  and hence
  \begin{equation}
    \lambda=O\left(\frac{\delta^{2}}{\|f^\ast\|_\mathcal{F}^2}{\left(\frac{\delta}{\|f^\ast\|_{\mathcal{F}}\sqrt{q_{\min}}}\right)}^{-\frac{2p}{1+p}}\right).
  \end{equation}
  We can fulfill~\eqref{seq:bound_T} with $N$ satisfying
  \begin{equation}
    N=O\left(\frac{\|f^\ast\|_\mathcal{F}^4\log\left(\frac{1}{\epsilon}\right)}{\delta^{4} q_{\min}^2}{\left(\frac{\|f^\ast\|_{\mathcal{F}}}{\lambda\delta\sqrt{q_{\min}}}\right)}^{\frac{4p}{1-p}}\right),
  \end{equation}
  which yields the conclusion.

  We remark that if we ignore factors of an arbitrarily small degree $O(p)$ for simplicity,
  our bounds reduce to
  \begin{align}
    \lambda&=O\left(\frac{\delta^{2}}{\|f^\ast\|_\mathcal{F}^2}\right),\\
    N&=O\left(\log\left(\frac{1}{\epsilon}\right)\times\frac{\|f^\ast\|_\mathcal{F}^4}{\delta^{4} q_{\min}^2}\right).
  \end{align}
  In this case, the excess classification error is bounded by
  \begin{equation}
    \mathbb{E}\left[\cR(\hat{f})-\cR^\ast\right]=O\left(\exp\left(-N\times\frac{\delta^4 q_{\min}^2}{\|f^\ast\|_\mathcal{F}^{4}}\right)\right).
  \end{equation}
  As a result, the exponential convergence of the excess classification error~\eqref{seq:expected_excess_testing_error} in terms of $N$ is obtained.
\end{proof}

\section{\label{sec:feasibility}A feasible implementation of the input model of the quantum algorithm in our setting}

We here clarify the runtime of the input model of the quantum algorithm of~\citet{NEURIPS2020_9ddb9dd5} in our setting of classification.
As the input model, this quantum algorithm uses the preparation of a quantum state
\begin{equation}
  \label{eq:state}
  \sum_x\sqrt{\hat{\rho}_\mathcal{X}(x)}\Ket{x}
\end{equation}
whose amplitude represents the square root of the empirical distribution $\hat{\rho}_\mathcal{X}$ of the $N_0$ unlabeled examples.
We here explain that the runtime of this input model can be bounded by $O(D\polylog(N_0))$ including the runtime of quantum random access memory (QRAM)~\citep{PhysRevA.78.052310,PhysRevLett.100.160501} used for the implementation.
Note that, in our setting of $N_0\approx \nicefrac{1}{\Delta^2}$ with precision $\Delta$ used for fixed- or floating-point number representation, poly-logarithmic overhead factors such as $O(\polylog(N_0))=O(\polylog(\nicefrac{1}{\Delta}))$ are ignored to simplify the presentation and clarify more essential factors.
The data structure used for the implementation of this input model is well defined in the work by~\citet{kerenidis_et_al:LIPIcs:2017:8154}.
It is known that QRAM used in this input model is implementable at a poly-logarithmic depth by a parallelized quantum circuit shown by~\citet{PhysRevA.78.052310,PhysRevLett.100.160501,PRXQuantum.2.020311}.
A detailed explanation on the feasible implementation of the input model has also been given in Sec.~B in Supplementary Material of the original paper of this quantum algorithm by~\citet{NEURIPS2020_9ddb9dd5}.
However, we here summarize these facts to avoid any potential confusion about feasibly of the quantum algorithm with QRAM in our setting.

The runtime of the quantum algorithm of~\citet{NEURIPS2020_9ddb9dd5} in terms of $N_0$ is determined by a step for preparing the quantum state~\eqref{eq:state}, using $N_0$ unlabeled examples in our setting.
To achieve this state preparation feasibly, along with collecting the $N_0$ unlabeled examples, we are to perform a preprocessing to count the number of unlabeled examples and store the empirical distribution in a sparse binary tree data structure proposed by~\citet{kerenidis_et_al:LIPIcs:2017:8154}.
In particular, using fixed-point number representation with precision $\Delta$, we represent the set $\mathbb{R}$ of real numbers as $\{0,\pm\Delta,\pm2\Delta,\ldots\}$ using $O(\log(\nicefrac{1}{\Delta}))$ bits.
The $D$-dimensional space $\mathcal{X}=\mathbb{R}^D$ of input data is represented using $O(D\log(\nicefrac{1}{\Delta}))$ bits, which divides $\mathcal{X}$ into $O(\nicefrac{1}{\Delta^D})$ grid regions.
Each leaf of the sparse binary tree counts and stores the number of unlabeled examples in each grid region, in such a way that the sparse binary tree does not store leaves with zero unlabeled example; that is, the number of leaves stored in the sparse binary tree is at most $N_0$.
Each parent in the binary tree stores the sum of the counts in its children, so that the root should store the number of all the unlabeled examples, i.e., $N_0$.
To construct this sparse binary tree, we increment the elements of the binary tree for each of the $N_0$ unlabeled examples, where each increment requires poly-logarithmic time~\citep{kerenidis_et_al:LIPIcs:2017:8154}.
Therefore, the runtime of this preprocessing is $\widetilde{O}(N_0)$, where $\widetilde{O}$ ignores the poly-logarithmic factors.
Remarkably, this runtime is the same scaling as just collecting the examples up to the poly-logarithmic factors.

We use this sparse binary tree with QRAM so that the empirical distribution can be input as amplitude of the quantum state~\eqref{eq:state}.
In particular, the preparation of this state is achieved by a parallelized quantum circuit at depth $O(D\log(\nicefrac{1}{\Delta}))$ to run a quantum algorithm of~\citet{G3}, where the QRAM is queried at each time step, in total $O(D\log(\nicefrac{1}{\Delta}))$ times~\citep{NEURIPS2020_9ddb9dd5,kerenidis_et_al:LIPIcs:2017:8154}.
Each of the queries to the QRAM is implementable within runtime $O(\polylog(N_0))$, as detailed below.
Thus, the overall runtime is bounded by
\begin{equation}
  O(D\log(\nicefrac{1}{\Delta})\times\polylog(N_0)).
\end{equation}
Since our analysis ignores the factor $\log(\nicefrac{1}{\Delta})$, this bound yields
\begin{equation}
  \label{seq:runtime}
  O(D\polylog(N_0)).
\end{equation}
This poly-logarithmic runtime in $N_0$ is achievable because the counts of the unlabeled examples are stored in the binary tree of height $O(D\log(\nicefrac{1}{\Delta}))$ to cover all the $O(\nicefrac{1}{\Delta^D})$ grid regions by its leaves, each of the $O(D\log(\nicefrac{1}{\Delta}))$ levels of this binary tree is used only in one of the $O(D\log(\nicefrac{1}{\Delta}))$ queries to QRAM\@, and each query to the QRAM is implemented by a $O(\polylog(N_0))$-depth quantum circuit to use $O(N_0)$ nodes at each level.
Importantly, the QRAM never measures and reads out classical bit values stored in the $O(N_0)$ nodes, which would require $O(N_0)$ runtime, but just performs $O(\polylog(N_0))$-depth unitary gates in parallel to maintain quantum superposition.

The QRAM is an architecture for using classical data in a quantum algorithm without destroying superposition, defined as (1) in the work by~\citet{PhysRevLett.100.160501}.
In our setting, the sparse binary tree has at most $N_0$ nodes as leaves, and each of the levels of the tree has at most $1,2,4,8,\ldots,N_0$ nodes, respectively, storing the classical data of the counting.
Each query to the QRAM performs a quantum circuit depending on one of the collections of these $1,2,4,8,\ldots,N_0$ nodes at each level.
The number of nodes to be used in each query is bounded by $O(N_0)$.
Then, the QRAM for using these $O(N_0)$ nodes is implementable by a $O(\polylog(N_0))$-depth parallelized quantum circuit on $O(N_0)$ qubits per query, e.g., given in Fig. 10 of the work by~\citet{PRXQuantum.2.020311}, which is combined with the above data structure to achieve the runtime~\eqref{seq:runtime}.

\vskip 0.2in

\bibliography{a}

\end{document}